\documentclass{lmcs} 


\usepackage{hyperref}
\usepackage{geometry}
\usepackage{lineno,hyperref}
\usepackage{amsmath,amssymb,amsfonts}
\usepackage{amsmath,amsthm}
\usepackage{color}
\usepackage{extarrows}
\theoremstyle{plain} 

\begin{document}

\title[]{Concurrent Matching Logic}

\author[S.~Wang]{ShangBei Wang}	
\address{Nanjing University of Aeronautics and Astronautics, Nanjing, China}	
\email{wangshangbei123@nuaa.edu.cn}  







\begin{abstract}
Matching logic cannot handle concurrency. We introduce concurrent matching logic (CML) to reason about fault-free
partial correctness of shared-memory concurrent programs. We also present a soundness proof for concurrent matching logic (CML) in terms of operational semantics. Under
certain assumptions, the assertion of CSL can be transformed into the assertion of CML. Hence, CSL can be seen as an
instance of CML.
\end{abstract}

\maketitle

\section*{Introduction}\label{S:one}
It is inevitable that the concurrent execution of shared-memory  programs produce races in which one process changes a piece of state that is simultaneously being used by another process. There are a number of approaches to guarantee both race-freedom and correctness when reasoning about shared-memory concurrent programs. Based on ``spatial separation", Hoare\cite{hoare1972towards} introduced formal proof rules for shared-memory concurrent programs. Expanding Hoare's work, Owicki and Gries\cite{owicki1976axiomatic}\cite{owicki1976verifying} introduced a syntax-directed logic for shared-memory concurrent programs. The critical variables and resources are key in Owicki and Gries's logic. The critical variables are identifiers which can by concurrently read and written by processes. Each occurrence of a critical variable must be inside a critical region protected by a resource name. As a result, processes are mutually exclusive access to critical variables and free of races. However, Brookes\cite{brookes2007semantics} pointed out that this approach works well for pointer-free shared-memory concurrent programs, but fails when pointer aliases are included in concurrent programs. Concurrent Separation Logic (CSL) was proposed by O'Hearn\cite{o2007resources} using separation logic\cite{reynolds2002separation}, together with an adaptation of the Owicki-Gries's methodology, to reason about partial correctness of concurrent pointer-programs. The innovation of CSL is to insert the separating conjunction into rules dealing with resource invariants and parallel composition. But Reynolds has shown that O'Hearn's rules are unsound without restrictions on resource invariants\cite{o2009separation}. Brookes\cite{brookes2007semantics} used action traces to provide a denotational semantics that can detect race. Then, the inference rules of O'Hearn were redefined in a more semantically way, and the sound of the proof rules was proved. Ian Wehrman and Josh Berdine found a counter example, suggesting that the sound of this method was based on the implicit assumption that no other program modified it. In order to avoid this problem, Brookes developed a fully compositional concurrent separation logic by adding ``rely set" to the assertions of CSL\cite{brookes2011revisionist}.\\
Dijkstra\cite{dijkstra1968cooperating} stated that one of the basic principles of concurrent program is that processes should be loosely connected. Processes should be considered completely independent of each other except when they are explicitly interacting with each other. The principle reflected in the above approaches is the idea of ``resource separation". At any time, the state can be divided into two separate portions, one for the process and the other satisfying the relevant resource invariant, for the available resource. When a process acquires a resource, it has ownership of the separate portion of state associated with the resource; When releasing a resource, it must ensure that resource invariant continues to hold and return ownership of the corresponding separate portion of state. The idea of ``resource separation" fit particularly well with the viewpoint of separation logic. Therefore, O'Hearn inserted the separating conjunction in appropriate places in the rules studied by Owicki and Gries, and then the very popular concurrent separation logic (CSL) came into being.\\
Matching logic\cite{rosu2017matching}\cite{rocsu2010matching} introduced by Grigore Ro\c{s}u has inherent support for heap separation without the need to extend the logic with separating conjunction. In other words, matching logic inherently supports the viewpoint of ``resource separation". In addition, Grigore Ro\c{s}u showed that separation logic is an instance of matching logic both syntactically and semantically\cite{rosu2017matching}. Therefore, it is natural to think that matching logic can be used for proving certain correction properties of concurrent programs. Unfortunately, matching logic does not typically handle concurrency. Inspired by the concurrent separation logic (CSL), we introduce Concurrent Matching Logic (CML) for reasoning about fault-free partial correctness of shared-memory concurrent programs in this paper.\\
First, we give an operational semantics model for concurrency, using K semantics framework\\\cite{rosu2010overview}\cite{rosu2017k}, which include race-detection. Our operational semantics , based on ``actions" ,is  transition traces semantics, describes the interleaving behavior of processes and without interference unless synchronized. Second, in order to be able to handle concurrency, we not only extend  Grigore Ro\c{s}u's matching logic inference rules, allowing resource declarations and concurrent compositions, but also adjust the pattern of the matching logic. Like CSL, the assertion requires a ``rely set" $A$, which represents a set of variables that are not changed by ``environment moves". CSL uses separation logic formulas to describe the state before and after process execution. Instead of logic formulas, matching logic uses patterns. The separation logic formula is abstract, and CSL can gracefully handle ``environment moves" simply by relying on the ``rely set". However, pattern involves ``low-level" operational aspects, such as how to express the state. In addition to a ``rely set" $A$, we need  a ``key set" $B$ to handle ``environment moves". ``Key set" B is also a collection of variables that is used when a concrete configuration $\gamma$ matches CML pattern. Finally, we give the notion of validity and prove CML is sound to our operational semantics model for concurrency.  We also analyze the relationship between the CML and the CSL, and point out that under certain assumptions, the assertions of CSL can be transformed into the assertions of CML. That is to say any property provable using CSL is also provable using CML.
\section{Preliminaries}
\subsection*{K Semantics Framework}
K is an executable semantic framework. The operational semantics of a program language $L$ is defined as a rewrite logic theory ($\Sigma_L$,$\mathcal{E}_L$,$\mathcal{R}_L$)\cite{marti1996rewriting} by K. Let $L^o$ be the algebraic specification $(\Sigma_L,\mathcal{E}_L^o)$ where $\mathcal{E}_L^o\subseteq \mathcal{E}_L$ and $\mathcal{T}^o$ be the initial $L^o$ algebra. Term $<\!\!s_c\!\!>_s<\!\!h_c\!\!>_h<\!\!N_c\!\!>_r$ in $\mathcal{T}^o$ is called state and $<\!\!k_c\!\!>_k<\!\!s_c\!\!>_s<\!\!h_c\!\!>_h<\!\!N_c\!\!>_r$ is called concrete configuration which is to add a sequence of commands to the state.\\
Let $v_1,v_2,\ldots,v_n$ are integer values and $l_1,l_2,\ldots,l_n$ are addresses which are also integer values.
A state$<\!\!s_c\!\!>_s<\!\!h_c\!\!>_h<\!\!N_c\!\!>_r$ consists of a store $<\!\!s_c\!\!>_s$, a heap $<\!\!h_c\!\!>_h$ and a set of resource names $<\!\!N_c\!\!>_r$. The store $<\!\!s_c\!\!>_s$ has the form of $<\!\!i_1\mapsto v_1,i_2\mapsto v_2,\ldots,i_n\mapsto v_n\!\!>_s$
mapping identifiers to integers. Let's define $\mathbf{dom}(s_c)=\{i| (i\mapsto v)\in s_c\}$. The heap $<\!\!h_c\!\!>_h$ maps addresses to integers and $\mathbf{dom}(h_c)=\{l| (l\mapsto v)\in h_c\}$. We use the notation $[s_c|i\mapsto v]$ to indicate that the store is consistent with all identifiers of $s_c$ except $i$, which is mapped to $v$; and the similar notation $[h_c|l\mapsto v']$. The notation $h_c\backslash l$ denotes the removal of address $l$ from the domain of $h_c$ and $\mathbf{dom}(h_c\backslash l)=\mathbf{dom}(h_c)-\{l\}$. The ``initial" state has the form of $<\!\!s_c\!\!>_s<\!\!h_c\!\!>_h<\!\!\{\}\!\!>_r$ because resources are initially available.\\
The rules in K are divided into structural rules and semantic rules. For example
$$e_1+e_2\rightleftharpoons e_1\curvearrowright \square + e_2$$
The symbol $\rightleftharpoons$ represent two structural rules, one from left to right and the other from
right to left. The first rule says that in the expression $e_1+e_2$, $e_1$ can be evaluated first, and $e_2$ is reserved for later evaluation. Since the above rules are bidirectional, when the first and second rules are used iteratively, they complete the evaluation of $e_1$ and $e_2$. Then, semantic rule tells how to process.
\subsection*{Separation Logic Formula} A separation logic formula\cite{reynolds2002separation}\cite{o2019separation} is given by the following grammar.
$$p::=b\;|\;\mathbf{emp}\;|\;e\mapsto e'\;|\;p_1*p_2\;|\;p_1\wedge p_2\;|\;p_1\vee p_2\;|\;\neg p$$
Let $<\!\!s_1\!\!>_s<\!\!h_1\!\!>_h<\!\!N_1\!\!>_r$, $<\!\!s_2\!\!>_s<\!\!h_2\!\!>_h<\!\!N_2\!\!>_r$ are states.
If $\mathbf{dom}(h_1)\cap \mathbf{dom}(h_2)=\emptyset$, $h_1$ and $h_2$ are called disjoint, written $h_1\bot h_2$. we also write $h_1\cdot h_2=h_1 \cup h_2$.\\
The satisfaction relation $<\!\!s_c\!\!>_s<\!\!h_c\!\!>_h<\!\!N_c\!\!>_r\models p$ is defined as follows:
\begin{align*}
&<\!\!s_{c}\!\!>_s<\!\!h_{c}\!\!>_h<\!\!N_{c}\!\!>_r \models b \xLeftrightarrow{def} s_{c}(b)=true \\
   &<\!\!s_{c}\!\!>_s<\!\!h_{c}\!\!>_h<\!\!N_{c}\!\!>_r \models \mathbf{emp} \xLeftrightarrow{def} \mathbf{dom}(h_{c})=\emptyset \\
   &<\!\!s_{c}\!\!>_s<\!\!h_{c}\!\!>_h<\!\!N_{c}\!\!>_r \models e\mapsto e' \xLeftrightarrow{def} (s_{c}(e)\in\mathbf{dom}(h_{c})) \wedge (h_{c}(s_{c}(e))=s_{c}(e'))\\
   &<\!\!s_c\!\!>_s<\!\!h_{c}\!\!>_h<\!\!N_{c}\!\!>_r \models p\wedge q \xLeftrightarrow{def} (<\!\!s_c\!\!>_s<\!\!h_{c}\!\!>_h<\!\!N_{c}\!\!>_r \models p)\; \wedge (<\!\!s_c\!\!>_s<\!\!h_{c}\!\!>_h<\!\!N_{c}\!\!>_r \models q)\\
   &<\!\!s_c\!\!>_s<\!\!h_{c}\!\!>_h<\!\!N_{c}\!\!>_r \models p\vee q \xLeftrightarrow{def} (<\!\!s_c\!\!>_s<\!\!h_{c}\!\!>_h<\!\!N_{c}\!\!>_r \models p)\; \vee (<\!\!s_c\!\!>_s<\!\!h_{c}\!\!>_h<\!\!N_{c}\!\!>_r \models q)\\
   &<\!\!s_c\!\!>_s<\!\!h_{c}\!\!>_h<\!\!N_{c}\!\!>_r \models \neg p \xLeftrightarrow{def} \neg <\!\!s_c\!\!>_s<\!\!h_{c}\!\!>_h<\!\!N_{c}\!\!>_r \models p\\
   &<\!\!s_c\!\!>_s<\!\!h_{c}\!\!>_h<\!\!N_{c}\!\!>_r \models p*q \xLeftrightarrow{def} \exists\; h_{c1}, h_{c2}. (h_{c1}\bot h_{c2}) \wedge (h=h_{c1}\cdot h_{c2}) \wedge\\
   &(<\!\!s_c\!\!>_s<\!\!h_{c1}\!\!>_h<\!\!N_{c}\!\!>_r \models p)\; \wedge (<\!\!s_c\!\!>_s<\!\!h_{c2}\!\!>_h<\!\!N_{c}\!\!>_r \models q)
\end{align*}
\begin{lem}
Let $p$ is a separation logic formula, $e$ is a expression, $i$ is an identifier that occurs freely in formula $p$, $p[e/i]$ means to replace $i$ in $p$ with $e$, then
$$<\!\!s_c\!\!>_s<\!\!h_{c}\!\!>_h<\!\!N_{c}\!\!>_r\models p[e/i]\Leftrightarrow <\!\![s_c|i\mapsto s_c(e)]\!\!>_s<\!\!h_{c}\!\!>_h<\!\!N_{c}\!\!>_r\models p$$
\end{lem}
\begin{proof}
by induction on the structure of $p$.
\end{proof}
A separation logic formula $p$ is precise\cite{brookes2007semantics} if, for all states $<\!\!s_c\!\!>_s<\!\!h_{c}\!\!>_h<\!\!N_{c}\!\!>_r$, there is at most one sub-heap $h_{c1}\subseteq h_{c}$ such that $<\!\!s_c\!\!>_s<\!\!h_{c1}\!\!>_h<\!\!N_{c}\!\!>_r\models p$.
\subsection*{Resource}
As in Owicki-Gries\cite{owicki1976verifying}, each resource name $r$ has a set $X$ containing the protected identifiers and a resource invariant $R$\cite{hoare1972towards}. A resource context $\Gamma$ has the form
$$r_1(X_1):R_1,\cdots,r_n(X_n):R_n$$
where $r_1,\cdots,r_n$ are different resource names, $R_1,\cdots,R_n$ are separation logic formulas and $r(X):R \in \Gamma$ imply $r$ protects $x\in X$. \\
Let $\mathbf{owned}(\Gamma)=X_1 \cup X_2 \cup \cdots \cup X_n$, $\mathbf{inv}(\Gamma)=R_1*R_2*\cdots*R_n$, and $\mathbf{dom}(\Gamma)=\{r_1,r_2,\cdots,r_n\}$.
For each $i$, if $R_i$ is precise and $\mathbf{free}(R_i)\subseteq X_i$, we say $\Gamma$ is well-formed.
\subsection*{Matching Logic Pattern}
In G. Ro\c{s}u's matching logic, program variables are syntactic constants. In other words, one cannot quantify over program variables. Let $\mathit{Var}$ is an infinite set of logical variables. The pattern of matching logic has the form
$\exists X((o=c)\land p)$
where "$\mathit{o}$" is a placeholder; $X \subset \mathit{Var}$ is a set of bound variables; $c$ is a pattern structure; $\mathit{FOL}_=$ formula $p$ is a constraint. Valuation $(\gamma,\tau)$ has a concrete configuration $\gamma$ and a map $\tau:\mathit{Var}\to Int$.
$(\gamma,\tau)\models \exists X((o=c)\land p)$ iff there exists $\theta_\tau:\mathit{Var}\to Int$ with $\theta_\tau\!\!\upharpoonright_{\mathit{Var}/X}=\tau\!\!\upharpoonright_{\mathit{Var}/X}$ such that $\gamma=\theta_\tau(c)$ and $\theta_\tau \models p$.\\
Matching logic works well on the sequential processes\cite{chen2019matching}\cite{chen2021matching}\cite{rocsu2012hoare}\cite{rosu2011matching}, but not on the shared-memory concurrent processes. The reasons are as follows:
\begin{itemize}
  \item Matching logic cannot handle concurrency. There are critical variables which can by concurrently read and written by processes. Matching logic has no rules to guarantee that processes are mutually exclusive access to the critical variables and freedom from races;
  \item Program variables are syntactic constants in matching logic. The syntax-directed method such as ``rely set" representing a set of program identifiers assumed to be left unmodified by the ``environment moves" cannot be used in matching logic;
  \item CSL uses separation logic formula to describe the state before and after process execution. Instead of separation logic formula, matching logic uses pattern. Separation logic formula is relatively abstract and CSL can deal with ``environment moves" gracefully. However, pattern involves ``low-level" operational aspects, such as how to express the state. ``Environment moves" is a disaster for matching logic.
\end{itemize}
\section{Syntax}
\subsection*{Syntax}
The program language $L$ is the same as in the Concurrent separation logic\cite{brookes2007semantics}. Suppose $r,i,e,b,E,c$ are
meta-variables, and $r$ represents resource names, $i$ represents identifiers, $e$ represents integer expressions, $b$ represents boolean expressions, $E$ represents list expressions, $c$ represents commands. The resource name acts like a binary semaphore. It is also an integer variable, but the value is limited to either 0, which means the resource is in used, or 1, which means the resource is available. The expression is pure, that is, the value of the expression is heap-independent. The syntax of command is defined as follows:
$$c::=\mathbf{skip}\mid i:=e\mid i:=[e] \mid [e_1]:=e_2\mid i:=\mathbf{cons}\;E \mid \mathbf{dispose}\;e\mid k_1;k_2$$
$$\mid\mathbf{if}\;b\;\mathbf{then}\;k_1\;\mathbf{else}\;k_2\mid\mathbf{while}\;b\;\mathbf{do}\;k\mid\mathbf{resource}\;r\;\mathbf{in}\;k\mid\mathbf{with}\;r\;\mathbf{when}\;b\;\mathbf{do}\;k\mid k_1||k_2$$
$\mathbf{with}\;r\;\mathbf{when}\;b\;\mathbf{do}\;k$ is a conditional critical region of resource $r$. Before a process enters the conditional critical section of resource $r$, it must wait until resource $r$ is available, then obtain resource $r$ and estimate the value of $b$: if $b$ is true, the process executes $k$ and releases resource $r$ when $k$ completes; if $b$ is false, the process releases resources $r$ and waits for a retry. A resource can only be held by one process at a time. $\mathbf{resource}\;r\;\mathbf{in}\;k$ introduces a local resource name $r$, whose scope is $k$, it means that the resource $r$ is assumed initially available in $k$ and the actions involving $r$ are executed without interference.
\subsection*{Actions}
An action is an atomic unit used to measure the execution of a program. Let $\lambda$ is a meta-variable ranging over actions and $\lambda$ has the following form
$$\lambda::=\delta\mid i=v\mid i:=v\mid [l]=v \mid [l]:=v \mid \mathbf{alloc}(l,E)\mid \mathbf{disp}\;l\mid \mathbf{try}\;r\mid \mathbf{acq}\;r \mid \mathbf{rel}\;r \mid \mathbf{abort}$$
where $v$ ranges over integers, $l$ over addresses which are also integers, $E$ over list of integers. Every action has a natural and intuitive explanation. For example, $\mathbf{try}\;r$ means that the resource named $r$ failed to be obtained.\\
Let $\mathbf{free}(\lambda)$ is the set of identifiers that occur freely in $\lambda$, $\mathbf{mod}(\lambda)$ is the set of identifiers that can be modified in $\lambda$, $\mathbf{writes}(\lambda)$ is the set of identifiers or heap cells that can be modified in $\lambda$, $\mathbf{reads}(\lambda)$ is the set of identifiers or heap cells whose values are read by $\lambda$, and $\mathbf{res}(\lambda)$ is the set of resource names that occur freely in $\lambda$.
\begin{equation*}
\mathbf{mod}(\lambda)=\left\{
             \begin{array}{lc}
             \{i\}, &  \lambda\equiv i:=v\\
             \emptyset, & otherwise
             \end{array}
\right.
\end{equation*}
\begin{equation*}
\mathbf{writes}(\lambda)=\left\{
             \begin{array}{lc}
             \{i\}, &  \lambda\equiv i:=v\\
             \{l\}, &  \lambda\equiv [l]:=v\\
             \{l,l+1,\cdots,l+n\}, &  \lambda\equiv \mathbf{alloc}(l,E)\\
             \{l\}, &  \lambda\equiv \mathbf{dispose}\;l\\
             \emptyset, & otherwise
             \end{array}
\right.
\end{equation*}
\begin{equation*}
\mathbf{reads}(\lambda)=\left\{
             \begin{array}{lc}
             \{i\}, &  \lambda\equiv i=v\\
             \{l\}, &  \lambda\equiv [l]=v\\
             \emptyset, & otherwise
             \end{array}
\right.
\end{equation*}
$$\mathbf{free}(\lambda)=\mathbf{reads}(\lambda)\cup \mathbf{writes}(\lambda)$$
$$\mathbf{mod}(\lambda)\subseteq\mathbf{writes}(\lambda)$$
\section*{Concurrent Matching Logic}
\subsection*{Concurrent Matching Logic Pattern}
A first major distinction between concurrent matching logic (CML) and G. Ro\c{s}u's matching logic is that program variables are logical variables. Instead of assuming $p$ is an arbitrary $\mathit{FOL}_=$ formula, we require $p$ to be separation logic formula and to be precise.\\
Let $\mathbf{Var}$ is a set of logical variables, $i_1,i_2,\ldots,i_n$ are identifiers, $v_1,v_2,\ldots,v_n$ are integer values, $l_1,l_2,\ldots,l_n$ are addresses and $x_1,x_2,\ldots,x_n \in \mathbf{Var}$. Intuitively, at any stage of program execution, the state can be divided into three portions, one portion owned by the program, one portion owned by the environment, and the rest portion belonging to currently available resources. The pattern of CML is used to describe the state before and after the program is executed. Therefore, the definition of CML pattern is as follows.
\begin{defi}
Concurrent matching logic(CML) pattern has the form
$$ \exists X((o=<\!\!k\!\!>_k<\!\!s\!\!>_s<\!\!h_1,h_2,H\!\!>_h<\!\!N_1,N_2\!\!>_r)\land p)$$
where $X\in \mathbf{Var}$ is the set of bound variables; $k$ is a sequence of commands; $s$ has the form of $i_1\mapsto x_1,\ldots, i_n\mapsto x_n$; $h_1,h_2,H$ are mappings from addresses to integers; $N_1,N_2$ are sets of resources; $p$ is the separation logic formula.
\end{defi}
Note that $h_1,h_2,H,N_1,N_2$ do not contain the variables in $\mathbf{Var}$. In CML pattern, $h_1,N_1$ represents part of the heap and part of the resource owned by the process, $h_2,N_2$ is owned by the environment, and $H$ represents the remaining heap and satisfies the resource invariants of the currently available resources. However, how do we divide store into corresponding  portions? This is impossible because store can be shared between concurrent processes. Therefore, we introduce the ``key set" B, which is a set of identifiers owned by the process. More importantly, identifiers in ``key set" B cannot be changed by ``environment moves".
\begin{defi}
A concrete configuration $\gamma=<\!\!k\!\!>_k<\!\!s_c\!\!>_s<\!\!h_{1c},h_{2c},H_c\!\!>_h<\!\!N_{1c},N_{2c}\!\!>_r$ matches CML pattern $\exists X((o=<\!\!k\!\!>_k<\!\!s\!\!>_s<\!\!h_1,h_2,H\!\!>_h<\!\!N_1,N_2\!\!>_r)\land p)$ on ``key set" $B$, iff there is some $\tau:\mathit{Var}\to Int$ such that $(\gamma,\tau)\models_{B} \exists X((o=<\!\!k\!\!>_k<\!\!s\!\!>_s<\!\!h_1,h_2,H\!\!>_h<\!\!N_1,N_2\!\!>_r)\land p)$ which is equivalent that
\begin{itemize}
  \item there exists some $\theta_\tau:\mathit{Var}\to Int$ with $\theta_\tau\!\!\upharpoonright_{\mathit{Var}/X}=\tau\!\!\upharpoonright_{\mathit{Var}/X}$;
  \item $s_c\upharpoonright_{B}=\theta_\tau(s)\upharpoonright_{B}$;
  \item $h_{1c}=h_1$ and $N_{1c} = N_1$;
  \item $<\!\!s_c\!\!>_s<\!\!h_{1c}\!\!>_h<\!\!\{\}\!\!>_r\models p$.
\end{itemize}
\end{defi}
\begin{defi}
Concurrent matching logic(CML) assertion has the form
$\Gamma \vdash_{A,B} \exists X((o=<\!\!k\!\!>_k<\!\!s\!\!>_s<\!\!h_1,h_2,H\!\!>_h<\!\!N_1,N_2\!\!>_r)\land p)\Downarrow \exists X'((o=<\!\!\cdot\!\!>_k<\!\!s'\!\!>_s<\!\!h_1',h_2',H'\!\!>_h<\!\!N_1',N_2'\!\!>_r)\land q)$
where $A$ is rely set;  $B$ is key set; $p$, $q$ and the resource invariants in $\Gamma$, do not mention resource names. We say that an assertion is well-formed if $\Gamma$ is well-formed resource context, $\mathbf{free}(p)\cup \mathbf{free}(q)\subseteq A$, $\mathbf{free}(k)\subseteq \mathbf{owned}(\Gamma)\cup A$, and $B\subseteq\mathbf{dom}(s)$.
\end{defi}
The following inference rules will restrict which identifiers $k$ can write and read. $k$ can only read and write the identifier protected by resource $r$ inside a critical region of $r$.
For the sake of convenience in writing, $<\!\!k\!\!>_k<\!\!s\!\!>_s<\!\!h_1,h_2,H\!\!>_h<\!\!N_1,N_2\!\!>_r\land p$ is abbreviated as $<\!\!k\!\!>_k LS\land p$ where $LS=<\!\!s\!\!>_s<\!\!h_1,h_2,H\!\!>_h<\!\!N_1,N_2\!\!>_r$. We also use the following notation convention:
\begin{align*}
  LS[s\mid i\mapsto v]&\;\;\;\; instead\;\; of\;\;\;\; <\!\![s\mid i\mapsto v]\!\!>_s<\!\!h_1,h_2,H\!\!>_h<\!\!N_1,N_2\!\!>_r\\
  LS[h_1\mid l\mapsto v]&\;\;\;\; instead\;\; of\;\;\;\; <\!\!s\!\!>_s<\!\![h_1\mid l\mapsto v],h_2,H\!\!>_h<\!\!N_1,N_2\!\!>_r\\
  LS[h_1\setminus l]&\;\;\;\; instead\;\; of\;\;\;\; <\!\!s\!\!>_s<\!\![h_1\setminus l],h_2,H\!\!>_h<\!\!N_1,N_2\!\!>_r\\
  LS[s\setminus Y]&\;\;\;\; instead\;\; of\;\;\;\; <\!\![s\setminus Y]\!\!>_s<\!\!h_1,h_2,H\!\!>_h<\!\!N_1,N_2\!\!>_r
\end{align*}
\subsection*{Inference rules}
\begin{itemize}
\item SKIP
  $$\dfrac{\cdot}{\Gamma \vdash_{A,B} \exists X(o=<\!\!\mathbf{skip}\!\!>_k LS \land p)\Downarrow \exists X(o=<\!\!\cdot\!\!>_k LS \land p)}$$
  if $\mathbf{free}(p)\subseteq A$
\item ASSIGNMENT
  $$\dfrac{\Gamma \vdash_{A,B} \exists X(o=<\!\!\rho_e\!\!>_k LS\land p[e/i])\Downarrow \exists X(o=<\!\!v\!\!>_k LS\land p[e/i])}{\Gamma \vdash_{A,B} \exists X(o=<\!\!i:=e\!\!>_k LS\land p[e/i])\Downarrow \exists X(o=<\!\!\cdot\!\!>_k LS[s\mid i\mapsto v]\land p)}$$
   if $i\notin \mathbf{owned}(\Gamma)$ and $\mathbf{free}(e)\subseteq A$
\item SEQUENCE
  \begin{align*}
   &\Gamma \vdash_{A_1,B} \exists X_1(o=<\!\!k_1\!\!>_k LS_1\land p_1)\Downarrow\exists X_2(o=<\!\!\cdot\!\!>_k LS_2 \land p_2)\\
   &\Gamma \vdash_{A_2,B} \exists X_2(o=<\!\!k_2\!\!>_k LS_2 \land p_2)\Downarrow\exists X_3(o=<\!\!\cdot\!\!>_k LS_3 \land p_3)\\
   &\overline{\Gamma \vdash_{A_1\cup A_2, B} \exists X_1 (o=<\!\!k_1;k_2\!\!>_k LS_1\land p_1)\Downarrow\exists X_3 (o=<\!\!\cdot\!\!>_k LS_3\land p_3)}
\end{align*}
where $LS_1=<\!\!s\!\!>_s<\!\!h_1,h_2,H\!\!>_h<\!\!N_1,N_2\!\!>_r$, $LS_2=<\!\!s'\!\!>_s<\!\!h_1',h_2,H'\!\!>_h<\!\!N_1',N_2\!\!>_r$,
$LS_3=<\!\!s''\!\!>_s<\!\!h_1'',h_2,H''\!\!>_h<\!\!N_1'',N_2\!\!>_r$
\item CONDITIONAL
\begin{align*}
   &\Gamma \vdash_{A,B} \exists X_1(o=<\!\!k_1\!\!>_k LS_1\land p\land b )\Downarrow\exists X_2(o=<\!\!\cdot\!\!>_k LS_2\land q)\\
   &\Gamma \vdash_{A,B} \exists X_1(o=<\!\!k_2\!\!>_k LS_1\land p\land \neg b)\Downarrow\exists X_2(o=<\!\!\cdot\!\!>_k LS_2\land q)\\
   &\overline{\Gamma \vdash_{A, B} \exists X_1 (o=<\!\!\mathbf{if}\;b\;\mathbf{then}\;k_1\;\mathbf{else}\;k_2\!\!>_k LS_1\land p)\Downarrow\exists X_2 (o=<\!\!\cdot\!\!>_k LS_2\land q)}
\end{align*}
where $LS_1=<\!\!s\!\!>_s<\!\!h_1,h_2,H\!\!>_h<\!\!N_1,N_2\!\!>_r$, $LS_2=<\!\!s'\!\!>_s<\!\!h_1',h_2,H'\!\!>_h<\!\!N_1',N_2\!\!>_r$
\item LOOP
$$\dfrac{\Gamma \vdash_{A,B} \exists X(o=<\!\!k\!\!>_k LS\land p \land b)\Downarrow \exists X(o=<\!\!\cdot\!\!>_k LS\land p)}{\Gamma \vdash_{A,B} \exists X(o=<\!\!\mathbf{while}\;b\;\mathbf{do}\;k\!\!>_k LS\land p)\Downarrow \exists X(o=<\!\!\cdot\!\!>_k LS\land p\land \neg b)}$$
\item PARALLEL
\begin{align*}
   &\Gamma \vdash_{A_1,B_1} \exists X_1(o=<\!\!k_1\!\!>_k LS_1\land p_1)\Downarrow\exists X_2(o=<\!\!\cdot\!\!>_k LS_1'\land q_1)\\
   &\Gamma \vdash_{A_2,B_2} \exists X_1(o=<\!\!k_2\!\!>_k LS_2\land p_2)\Downarrow\exists X_2(o=<\!\!\cdot\!\!>_k LS_2'\land q_2)\\
   &\overline{\Gamma \vdash_{A_1\cup A_2,B_1\cup B_2} \exists X_1 (o=<\!\!k_1 \parallel k_2\!\!>_k LS\land (p_1*p_2))\Downarrow\exists X_2 (o=<\!\!\cdot\!\!>_k LS_f\land (q_1* q_2))}
\end{align*}
where $LS=<\!\!s\!\!>_s<\!\!h_1\cdot h_2, h_3,H\!\!>_h<\!\!N_1\cup N_2, N_3\!\!>_r$ ,$LS_1=<\!\!s\!\!>_s<\!\!h_1,h_2\cdot h_3,H\!\!>_h<\!\!N_1,N_2\cup N_3\!\!>_r$, $LS_1'=<\!\!s'\!\!>_s<\!\!h_1',h_2'\cdot h_3',H'\!\!>_h<\!\!N_1',N_2'\cup N_3'\!\!>_r$, $LS_2=<\!\!s\!\!>_s<\!\!h_2,h_1\cdot h_3,H\!\!>_h<\!\!N_2,N_1\cup N_3\!\!>_r$,$LS_2'=<\!\!s'\!\!>_s<\!\!h_2',h_1'\cdot h_3',H'\!\!>_h<\!\!N_2',N_1'\cup N_3'\!\!>_r$,
$LS_f=<\!\!s'\!\!>_s<\!\!h_1'\cdot h_2',h_3',H'\!\!>_h<\!\!N_1'\cup N_2',N_3'\!\!>_r$
and $\mathbf{mod}(k_1) \cap A_2=\mathbf{mod}(k_1) \cap B_2=\mathbf{mod}(k_2) \cap A_1 =\mathbf{mod}(k_2) \cap B_1 = \emptyset$ and $N_1\cap N_2=\emptyset$
\item ENVIRONMENT MOVES
$$\dfrac{\Gamma \vdash_{A_2,B_2} \exists X(o=<\!\!k_2\!\!>_k LS_2\land p_2)\Downarrow \exists X'(o=<\!\!\cdot\!\!>_k LS_2'\land q_2)}{\Gamma \vdash_{A_1,B_1} \exists X'(o=<\!\!k_1\!\!>_k LS_1\land p_1)\Downarrow \exists X'(o=<\!\!k_1\!\!>_k LS_1'\land p_1)}$$
where $LS_1=<\!\!s\!\!>_s<\!\!h_1,h_2,H\!\!>_h<\!\!N_1,N_2\!\!>_r$, $LS_1'=<\!\!s'\!\!>_s<\!\!h_1,h_2',H'\!\!>_h<\!\!N_1,N_2'\!\!>_r$, $LS_2=<\!\!s\!\!>_s<\!\!h_2,h_1,H\!\!>_h<\!\!N_2,N_1\!\!>_r$,$LS_2'=<\!\!s'\!\!>_s<\!\!h_2',h_1,H'\!\!>_h<\!\!N_2',N_1\!\!>_r$
and $\mathbf{writes}(k_2) \cap A_1 =\mathbf{writes}(k_2) \cap B_1 = \emptyset$
\item REGION
$$\dfrac{\Gamma \vdash_{A\cup Y,B} \exists X(o=<\!\!k\!\!>_k LS_1\land(( p \land b)*R))\Downarrow \exists X'(o=<\!\!\cdot\!\!>_k LS_2\land (q*R))}{\Gamma,r(Y):R \vdash_{A,B} \exists X(o=<\!\!\mathbf{with}\;r\;\mathbf{when}\;b\;\mathbf{do}\;k\!\!>_k LS_1\land p)\Downarrow \exists X'(o=<\!\!\cdot\!\!>_k LS_2\land q)}$$
where $LS_1=<\!\!s\!\!>_s<\!\!h_1,h_2,H\!\!>_h<\!\!N_1,N_2\!\!>_r$, $LS_2=<\!\!s'\!\!>_s<\!\!h_1',h_2',H'\!\!>_h<\!\!N_1',N_2'\!\!>_r$
\item RESOURCE
$$\dfrac{\Gamma,r(Y):R \vdash_{A,B} \exists X(o=<\!\!k\!\!>_k LS_1'\land p )\Downarrow \exists X'(o=<\!\!\cdot\!\!>_k LS_2'\land q)}{\Gamma \vdash_{A\cup Y,B} \exists X(o=<\!\!\mathbf{resource}\;r\;\mathbf{in}\;k\!\!>_k LS_1\land (p*R))\Downarrow \exists X'(o=<\!\!\cdot\!\!>_k LS_2\land (q*R))}$$
where $LS_1=<\!\!s\!\!>_s<\!\!h_1\cdot h,h_2,H\!\!>_h<\!\!N_1,N_2\!\!>_r$, $LS_2=<\!\!s'\!\!>_s<\!\!h_1'\cdot h',h_2',H'\!\!>_h<\!\!N_1',N_2'\!\!>_r$, $LS_1'=<\!\!s\!\!>_s<\!\!h_1,h_2,H\cdot h\!\!>_h<\!\!N_1,N_2\!\!>_r$, $LS_2'=<\!\!s'\!\!>_s<\!\!h_1',h_2',H'\cdot h'\!\!>_h<\!\!N_1',N_2'\!\!>_r$.
\item LOOKUP
\begin{align*}
   &\Gamma \vdash_{A,B} \exists X(o=<\!\!\rho_e\!\!>_k LS\land p)\Downarrow \exists X(o=<\!\!l\!\!>_k LS\land p)\\
   &\Gamma \vdash_{A,B} \exists X(o=<\!\![l]=v\!\!>_k LS\land p)\Downarrow\exists X(o=<\!\!v\!\!>_k LS\land p)\\
   &\overline{\Gamma \vdash_{A,B} \exists X (o=<\!\!i:=[e]\!\!>_k LS\land p[v/i])\Downarrow\exists X (o=<\!\!\cdot\!\!>_k LS[s\mid i\mapsto v]\land p)}
\end{align*}
if
$i\notin \mathbf{owned}(\Gamma)$ and $\mathbf{free}(e)\subseteq A$
\item UPDATE
\begin{align*}
   &\Gamma \vdash_{A,B} \exists X(o=<\!\!\rho_e\!\!>_k LS\land p)\Downarrow \exists X(o=<\!\!l\!\!>_k LS\land p)\\
   &\Gamma \vdash_{A,B} \exists X(o=<\!\!\rho_{e'}\!\!>_k LS\land p)\Downarrow\exists X(o=<\!\!v\!\!>_k LS\land p)\\
   &\overline{\Gamma \vdash_{A,B} \exists X (o=<\!\![e]:=e'\!\!>_k LS\land p)\Downarrow\exists X (o=<\!\!\cdot\!\!>_k LS[h_1\mid l\mapsto v]\land p)}
\end{align*}
\item ALLOCATION
  $$\dfrac{\cdot}{\Gamma \vdash_{A,B} \exists X(o\!=<\!\!\mathbf{alloc}(l,[v_0,\!\cdots\!,v_n])\!\!>_k\! LS\!\land\! p)\!\Downarrow \!\exists X(o\!=<\!\!l\!\!>_k \!LS[h_1 \mid l\mapsto v_0,\!\cdots\!,l\!+\!n\!\mapsto\!v_n]\!\land \!p)}$$
\item DISPOSAL
$$\dfrac{\Gamma \vdash_{A,B} \exists X(o=<\!\!\rho_e\!\!>_k LS\land p)\Downarrow \exists X(o=<\!\!l\!\!>_k LS\land p)}{\Gamma \vdash_{A,B} \exists X(o=<\!\!\mathbf{dispose}\;e\!\!>_k LS\land p)\Downarrow \exists X(o=<\!\!\cdot\!\!>_kLS[h_1-\{l\mapsto-\}]\land p)}$$
\item FRAME
$$\dfrac{\Gamma \vdash_{A,B} \exists X(o=<\!\!k\!\!>_k LS_1\land p)\Downarrow \exists X'(o=<\!\!\cdot\!\!>_k LS_2\land q)}{\Gamma \vdash_{A\cup \mathbf{free}(R),B} \exists X(o=<\!\!k\!\!>_k LS_1\land (p*R))\Downarrow \exists X'(o=<\!\!\cdot\!\!>_k LS_2\land (q*R))}$$
where $LS_1=<\!\!s\!\!>_s<\!\!h_1,h_2,H\!\!>_h<\!\!N_1,N_2\!\!>_r$, $LS_2=<\!\!s'\!\!>_s<\!\!h_1',h_2',H'\!\!>_h<\!\!N_1',N_2'\!\!>_r$ and $\mathbf{mod}(c)\cap \mathbf{free}(R)=\emptyset$
\item CONSEQUENCE
$$\dfrac{\Gamma \vdash_{A,B} \exists X(o=<\!\!k\!\!>_k LS_1\land p)\Downarrow \exists X'(o=<\!\!\cdot\!\!>_k LS_2\land q)}{\Gamma \vdash_{A',B} \exists X(o=<\!\!k\!\!>_k LS_1\land p')\Downarrow \exists X'(o=<\!\!\cdot\!\!>_k LS_2\land q')}$$
where $LS_1=<\!\!s\!\!>_s<\!\!h_1,h_2,H\!\!>_h<\!\!N_1,N_2\!\!>_r$, $LS_2=<\!\!s'\!\!>_s<\!\!h_1',h_2',H'\!\!>_h<\!\!N_1',N_2'\!\!>_r$
and $A\subseteq A'$, $\exists X(o=<\!\!k\!\!>_k LS_1\land p')\Rightarrow \exists X(o=<\!\!k\!\!>_k LS_1\land p)$, $\exists X'(o=<\!\!\cdot\!\!>_k LS_2\land q)\Rightarrow \exists X'(o=<\!\!\cdot\!\!>_k LS_2\land q')$
\item AUXILIARY
$$\dfrac{\Gamma \vdash_{A\cup Y,B\cup Y} \exists X(o=<\!\!k\!\!>_k LS_1\land p)\Downarrow \exists X'(o=<\!\!\cdot\!\!>_k LS_2\land q)}{\Gamma \vdash_{A,B} \exists X(o=<\!\!k\setminus Y \!\!>_k LS_1[s\setminus Y]\land p)\Downarrow \exists X'(o=<\!\!\cdot\!\!>_k LS_2[s\setminus Y]\land q)}$$
where $LS_1=<\!\!s\!\!>_s<\!\!h_1,h_2,H\!\!>_h<\!\!N_1,N_2\!\!>_r$, $LS_2=<\!\!s'\!\!>_s<\!\!h_1',h_2',H'\!\!>_h<\!\!N_1',N_2'\!\!>_r$
if $Y$ is a auxiliary for $k$, $Y\cap (\mathbf{free}(p)\cup \mathbf{free}(q))=\emptyset$, $Y\cap \mathbf{owned}(\Gamma)=\emptyset$
\end{itemize}
\subsection*{Examples}
We now discuss some example programs and assertions, to illustrate the way the inference rules work.
\begin{enumerate}
\item \emph{The rely set}\\
The assertions\\
(a)\;\;$r(\{a,x\}):x=a \land \mathbf{emp}\vdash_{\{a,t\},\{x,t\}}\exists \{t_v,x_v,a_v\}(o=<\!\!\mathbf{with}\;r\;\mathbf{when}\;true\;\mathbf{do}\;t:=x\!\!>_k<\!\!x\mapsto x_v,t\mapsto t_v,a\mapsto a_v\!\!>_s<\!\!\cdot\!\!>_h<\!\!\cdot\!\!>_r\land \mathbf{emp})\Downarrow\exists \{t_v,x_v,a_v\}(o=<\!\!\cdot\!\!>_k<\!\!x\mapsto x_v,t\mapsto x_v, a\mapsto a_v\!\!>_s<\!\!\cdot\!\!>_h<\!\!\cdot\!\!>_r\land t=a\land \mathbf{emp})$\\
  (b)\;\;$r(\{a,x\}):x=a \land \mathbf{emp}\vdash_{\{a,t\},\{x,t\}}\exists \{t_v,x_v,a_v\}(o=<\!\!\mathbf{with}\;r\;\mathbf{when}\;true\;\mathbf{do}\;x:=t\!\!>_k<\!\!x\mapsto x_v,t\mapsto x_v,a\mapsto a_v\!\!>_s<\!\!\cdot\!\!>_h<\!\!\cdot\!\!>_r\land t=a\land \mathbf{emp})\Downarrow\exists \{t_v,x_v,a_v\}(o=<\!\!\cdot\!\!>_k<\!\!x\mapsto x_v,t\mapsto x_v,a\mapsto a_v\!\!>_s<\!\!\cdot\!\!>_h<\!\!\cdot\!\!>_r\land \mathbf{emp})$\\
are valid and well-formed. Each is also provable from ASSIGNMENT, CONSEQUENCE and REGION.\\
  Let $c_1$ be $\mathbf{with}\;r\;\mathbf{when}\;true\;\mathbf{do}\;t:=x; \mathbf{with}\;r\;\mathbf{when}\;true\;\mathbf{do}\;x:=t$. The assertion\\
  (c)\;\;$r(\{a,x\}):x=a \land \mathbf{emp}\vdash_{\{a,t\},\{x,t\}}\exists \{t_v,x_v,a_v\}(o=<\!\!c_1\!\!>_k<\!\!x\mapsto x_v,t\mapsto t_v, a\mapsto a_v\!\!>_s<\!\!\cdot\!\!>_h<\!\!\cdot\!\!>_r\land \mathbf{emp})\Downarrow\exists \{t_v,x_v,a_v\}(o=<\!\!\cdot\!\!>_k<\!\!x\mapsto x_v,t\mapsto x_v, a\mapsto a_v\!\!>_s<\!\!\cdot\!\!>_h<\!\!\cdot\!\!>_r\land \mathbf{emp})$\\
  is valid, well-formed, and provable from (a) and (b) using SEQUENCE. Let $c_1$ be as above and let $c_2$ be $\mathbf{with}\;r\;\mathbf{when}\;true\;\mathbf{do}\;(x:=x+1;a:=a+1)$. The assertion\\
(d)\;\;$r(\{a,x\}):x=a \land \mathbf{emp}\vdash_{A,\{x,t,a\}}\exists \{t_v,x_v,a_v\}(o=<\!\!c_1\parallel c_2\!\!>_k<\!\!x\mapsto x_v,t\mapsto t_v, a\mapsto a_v\!\!>_s<\!\!\cdot\!\!>_h<\!\!\cdot\!\!>_r\land \mathbf{emp})\Downarrow\exists \{t_v',x_v',a_v'\}(o=<\!\!\cdot\!\!>_k<\!\!x\mapsto x_v',t\mapsto t_v', a\mapsto a_v'\!\!>_s<\!\!\cdot\!\!>_h<\!\!\cdot\!\!>_r\land \mathbf{emp})$\\
  is not valid because of  executing $c_1 \parallel c_2$ without interference does not necessarily preserve equality of $x$ and $a$. Moreover, the assertion (d) is not provable. Suppose\\
  $r(\{a,x\}):x=a \land \mathbf{emp}\vdash_{A_1,B_1}\exists \{t_v,x_v,a_v\}(o=<\!\!c_1\!\!>_k<\!\!x\mapsto x_v,t\mapsto t_v,a\mapsto a_v\!\!>_s<\!\!\cdot\!\!>_h<\!\!\cdot\!\!>_r\land \mathbf{emp})\Downarrow\exists \{t_v',x_v',a_v'\}(o=<\!\!\cdot\!\!>_k<\!\!x\mapsto x_v',t\mapsto t_v', a\mapsto a_v'\!\!>_s<\!\!\cdot\!\!>_h<\!\!\cdot\!\!>_r\land \mathbf{emp})$\\
  $r(\{a,x\}):x=a \land \mathbf{emp}\vdash_{A_2,B_2}\exists \{t_v,x_v,a_v\}(o=<\!\!c_2\!\!>_k<\!\!x\mapsto x_v,t\mapsto t_v,a\mapsto a_v\!\!>_s<\!\!\cdot\!\!>_h<\!\!\cdot\!\!>_r\land \mathbf{emp})\Downarrow\exists \{t_v',x_v',a_v'\}(o=<\!\!\cdot\!\!>_k<\!\!x\mapsto x_v',t\mapsto t_v',a\mapsto a_v'\!\!>_s<\!\!\cdot\!\!>_h<\!\!\cdot\!\!>_r\land \mathbf{emp})$\\
  $A_1$ would have to contain $x$ or $a$, but $c_2$ modifies both of these variables, so the side condition on the PARALLEL rule would fail. This examples is the counterexample found by Ian Wehrman and Josh Berdine showing that, without the rely set, the assertion (d) is is provable in the original concurrent separation logic but not valid.
\item \emph{Auxiliary variable}\\
   $r(\{x,a,b\}):x=a+b \land\mathbf{emp} \vdash_{\{a\},\{a\}}\exists \{x_v,a_v,b_v\}(o=<\!\!\mathbf{with}\;r\;\mathbf{when}\;true\;\mathbf{do}\;x:=x+1;a:=a+1\!\!>_k<\!\!x\mapsto x_v,a\mapsto a_v,b\mapsto b_v\!\!>_s<\!\!\cdot\!\!>_h<\!\!\cdot\!\!>_r\land a=0\land\mathbf{emp})\Downarrow\exists \{x_v,a_v,b_v\}(o=<\!\!\cdot\!\!>_k<\!\!x\mapsto x_v+2,a\mapsto a_v+1,b\mapsto b_v+1\!\!>_s<\!\!\cdot\!\!>_h<\!\!\cdot\!\!>_r\land a=1\land\mathbf{emp})$\\
  is valid, and provable from REGION and ENVIRONMENT MOVE, because\\
$\vdash_{\{x,a,b\},\{a\}}\exists \{x_v,a_v,b_v\}(o=<\!\!x:=x+1;a:=a+1\!\!>_k<\!\!x\mapsto x_v,a\mapsto a_v,b\mapsto b_v\!\!>_s<\!\!\cdot\!\!>_h<\!\!\cdot\!\!>_r\land (x=a+b\land \mathbf{emp})*(a=0\land\mathbf{emp}))\Downarrow\exists \{x_v,a_v,b_v\}(o=<\!\!\cdot\!\!>_k<\!\!x\mapsto x_v+1,a\mapsto a_v+1,b\mapsto b_v\!\!>_s<\!\!\cdot\!\!>_h<\!\!\cdot\!\!>_r\land (x=a+b\land \mathbf{emp})*(a=1\land\mathbf{emp}))$\\
  is provable from ASSIGNMENT, SEQUENCE and  CONSEQUENCE.\\
 Similarly we can prove\\
 $r(\{x,a,b\}):x=a+b \land\mathbf{emp} \vdash_{\{b\},\{b\}}\exists \{x_v,a_v,b_v\}(o=<\!\!\mathbf{with}\;r\;\mathbf{when}\;true\;\mathbf{do}\;x:=x+1;b:=b+1\!\!>_k<\!\!x\mapsto x_v,a\mapsto a_v,b\mapsto b_v\!\!>_s<\!\!\cdot\!\!>_h<\!\!\cdot\!\!>_r\land b=0\land\mathbf{emp})\Downarrow\exists \{x_v,a_v,b_v\}(o=<\!\!\cdot\!\!>_k<\!\!x\mapsto x_v+2,a\mapsto a_v+1,b\mapsto b_v+1\!\!>_s<\!\!\cdot\!\!>_h<\!\!\cdot\!\!>_r\land b=1\land\mathbf{emp})$\\
  Using PARALLEL and CONSEQUENCE we can then derive\\
  $r(\{x,a,b\}):x=a+b \land\mathbf{emp} \vdash_{\{a,b\},\{a,b\}}\exists \{x_v,a_v,b_v\}(o=<\!\!(\mathbf{with}\;r\;\mathbf{when}\;true\;\mathbf{do}\;x:=x+1;a:=a+1)\parallel(\mathbf{with}\;r\;\mathbf{when}\;true\;\mathbf{do}\;x:=x+1;b:=b+1)\!\!>_k<\!\!x\mapsto x_v,a\mapsto a_v,b\mapsto b_v\!\!>_s<\!\!\cdot\!\!>_h<\!\!\cdot\!\!>_r\land a=0\land b=0\land\mathbf{emp})\Downarrow\exists \{x_v,a_v,b_v\}(o=<\!\!\cdot\!\!>_k<\!\!x\mapsto x_v+2,a\mapsto a_v+1,b\mapsto b_v+1\!\!>_s<\!\!\cdot\!\!>_h<\!\!\cdot\!\!>_r\land a=1\land b=1\land\mathbf{emp})$\\
  which is also valid. Using RESOURCE and CONSEQUENCE we then obtain\\
  $\vdash_{\{x,a,b\},\{a,b\}}\exists \{x_v,a_v,b_v\}(o=<\!\!\mathbf{resource}\;r\;\mathbf{in}\;((\mathbf{with}\;r\;\mathbf{when}\;true\;\mathbf{do}\;x:=x+1;a:=a+1)\parallel(\mathbf{with}\;r\;\mathbf{when}\;true\;\mathbf{do}\;x:=x+1;b:=b+1))\!\!>_k<\!\!x\mapsto x_v,a\mapsto a_v,b\mapsto b_v\!\!>_s<\!\!\cdot\!\!>_h<\!\!\cdot\!\!>_r\land x=a+b \land a=0\land b=0\land\mathbf{emp})\Downarrow\exists \{x_v,a_v,b_v\}(o=<\!\!\cdot\!\!>_k<\!\!x\mapsto x_v+2,a\mapsto a_v+1,b\mapsto b_v+1\!\!>_s<\!\!\cdot\!\!>_h<\!\!\cdot\!\!>_r\land x=a+b\land a=1\land b=1\land\mathbf{emp})$\\
  By ASSIGNMENT, SEQUENCE and CONSEQUENCE we then have\\
  $\vdash_{\{x,a,b\},\{a,b\}}\exists \{x_v,a_v,b_v\}(o=<\!\!a:=0\; ; b:=0;\;\mathbf{resource}\;r\;\mathbf{in}\;((\mathbf{with}\;r\;\mathbf{when}\;true\;\mathbf{do}$\\$x:=x+1;a:=a+1)\parallel(\mathbf{with}\;r\;\mathbf{when}\;true\;\mathbf{do}\;x:=x+1;b:=b+1))\!\!>_k<\!\!x\mapsto x_v,a\mapsto a_v,b\mapsto b_v\!\!>_s<\!\!\cdot\!\!>_h<\!\!\cdot\!\!>_r\land x=0\land\mathbf{emp})\Downarrow\exists \{x_v,a_v,b_v\}(o=<\!\!\cdot\!\!>_k<\!\!x\mapsto x_v+2,a\mapsto a_v+1,b\mapsto b_v+1\!\!>_s<\!\!\cdot\!\!>_h<\!\!\cdot\!\!>_r\land x=2\land\mathbf{emp})$\\
  Finally, since $a,b$ is an auxiliary variable set for this program, and $\{a,b\}\cap \{x\}=\emptyset$ and $\{a,b\}\cap \emptyset=\emptyset$,  we can use the AUXILIARY rule to obtain\\
  $\vdash_{\{x\},\{\}}\exists \{x_v,a_v,b_v\}(o=<\!\!\mathbf{resource}\;r\;\mathbf{in}\;((\mathbf{with}\;r\;\mathbf{when}\;true\;\mathbf{do}\;x:=x+1)\parallel(\mathbf{with}\;r\;\mathbf{when}\; true\;\mathbf{do}\;x:=x+1))\!\!>_k<\!\!x\mapsto x_v\!\!>_s<\!\!\cdot\!\!>_h<\!\!\cdot\!\!>_r\land x=0\land\mathbf{emp})\Downarrow\exists \{x_v,a_v,b_v\}(o=<\!\!\cdot\!\!>_k<\!\!x\mapsto x_v+2\!\!>_s<\!\!\cdot\!\!>_h<\!\!\cdot\!\!>_r\land x=2\land\mathbf{emp})$
\end{enumerate}
\section{Semantic}
\subsection*{Local state}
From the perspective of the interaction between a process and its environment, A state $<\!\!s\!\!>_s<\!\!h\!\!>_h<\!\!N\!\!>_r$
can be expressed as $<\!\!s\!\!>_s<\!\!h_1,h_2,H\!\!>_h<\!\!N_1,N_2\!\!>_r$, where $h_1,N_1$ represents part of the heap and part of the resource owned by the process, $h_2,N_2$ is owned by the environment, and $H$ represents the remaining heap and satisfies the resource invariants of the currently available resources. Obviously, $h=h_1\cdot h_2 \cdot H$ and $N=N_1\cup N_2$.
For a well-formed resource context $\Gamma$ and a set of resource names $N$, let $\Gamma\upharpoonright N=\{ r(X):R \in \Gamma | r \in N \}$ and $\Gamma \setminus N=\{ r(X):R \in \Gamma | r \notin N \}$.
\begin{defi}
We say $<\!\!s\!\!>_s<\!\!h_1,h_2,H\!\!>_h<\!\!N_1,N_2\!\!>_r$ is a local state for $\Gamma$ if:
\begin{itemize}
  \item $h_1\bot h_2$, $h_1\bot H$ and $h_2\bot H$;
  \item $N_1 \cap N_2=\emptyset$, $s(r)=0 $ for $r\in N_1\cup N_2$, $s(r)=1 $ otherwise;
  \item $<\!\!s\!\!>_s<\!\!H\!\!>_h<\!\!N_1,N_2\!\!>_r\models \mathbf{inv}(\Gamma \setminus (N_1\cup N_2))$
\end{itemize}
\end{defi}
Let $\Sigma_\Gamma$ be the set of local states for $\Gamma$.
\subsection*{Semantic rules for Actions}
Let$<\!\!s\!\!>_s<\!\!h_1,h_2,H\!\!>_h<\!\!N_1,N_2\!\!>_r\in \Sigma_\Gamma$, $A$ is rely set, $B$ is key set, the semantic rules for actions are as follows:
\begin{itemize}
\item $<\!\!\delta\!\!>_k<\!\!s\!\!>_s<\!\!h_1,h_2,H\!\!>_h<\!\!N_1,N_2\!\!>_r\xrightarrow[\Gamma, A,B]{}<\!\!\cdot\!\!>_k<\!\!s\!\!>_s<\!\!h_1,h_2,H\!\!>_h<\!\!N_1,N_2\!\!>_r$
\item $<\!\!i=v\!\!>_k<\!\!s\!\!>_s<\!\!h_1,h_2,H\!\!>_h<\!\!N_1,N_2\!\!>_r\xrightarrow[\Gamma, A,B]{}<\!\!v\!\!>_k<\!\!s\!\!>_s<\!\!h_1,h_2,H\!\!>_h<\!\!N_1,N_2\!\!>_r$\\
   if $s(i)=v$\;\; and\;\; $i\in \mathbf{owend}(\Gamma\upharpoonright N_1)\cup A$
\item $<\!\!i=v\!\!>_k<\!\!s\!\!>_s<\!\!h_1,h_2,H\!\!>_h<\!\!N_1,N_2\!\!>_r\xrightarrow[\Gamma, A,B]{}\mathbf{abort}$ \;\;
   if $i\notin \mathbf{owend}(\Gamma\upharpoonright N_1)\cup A$
\item $<\!\![l]=v\!\!>_k<\!\!s\!\!>_s<\!\!h_1,h_2,H\!\!>_h<\!\!N_1,N_2\!\!>_r\xrightarrow[\Gamma, A, B]{}<\!\!v\!\!>_k<\!\!s\!\!>_s<\!\!h_1,h_2,H\!\!>_h<\!\!N_1,N_2\!\!>_r$\\
if $h_1(l)=v$
\item $<\!\![l]=v\!\!>_k<\!\!s\!\!>_s<\!\!h_1,h_2,H\!\!>_h<\!\!N_1,N_2\!\!>_r\xrightarrow[\Gamma, A, B]{}\mathbf{abort}$\;\;
if $l\notin \mathbf{dom}(h_1)$
\item $<\!\!i:=v\!\!>_k<\!\!s\!\!>_s<\!\!h_1,h_2,H\!\!>_h<\!\!N_1,N_2\!\!>_r\xrightarrow[\Gamma, A, B]{}<\!\!\cdot\!\!>_k<\!\![s\mid i\mapsto v]\!\!>_s<\!\!h_1,h_2,H\!\!>_h<\!\!N_1,N_2\!\!>_r$\\
if $i\notin \mathbf{owned}(\Gamma \setminus N_1)$
\item $<\!\!i:=v\!\!>_k<\!\!s\!\!>_s<\!\!h_1,h_2,H\!\!>_h<\!\!N_1,N_2\!\!>_r\xrightarrow[\Gamma, A, B]{} \mathbf{abort}$\;\;
if $i\in \mathbf{owned}(\Gamma \setminus N_1)$
\item $<\!\![l]:=v\!\!>_k<\!\!s\!\!>_s<\!\!h_1,h_2,H\!\!>_h<\!\!N_1,N_2\!\!>_r\xrightarrow[\Gamma, A, B]{}<\!\!\cdot\!\!>_k<\!\!s\!\!>_s<\!\![h_1\mid l\mapsto v],h_2,H\!\!>_h<\!\!N_1,N_2\!\!>_r$\;\;\;\;\;\;
if $l\in \mathbf{dom}(h_1)$
\item $<\!\![l]:=v\!\!>_k<\!\!s\!\!>_s<\!\!h_1,h_2,H\!\!>_h<\!\!N_1,N_2\!\!>_r\xrightarrow[\Gamma, A, B]{}\mathbf{abort}$ \;\;
if $l\notin \mathbf{dom}(h_1)$
\item $<\!\!\mathbf{alloc}(l,[v_0,\cdots,v_n])\!\!>_k<\!\!s\!\!>_s<\!\!h_1,h_2,H\!\!>_h<\!\!N_1,N_2\!\!>_r\xrightarrow[\Gamma, A, B]{}<\!\!l\!\!>_k<\!\!s\!\!>_s<\!\!h_1',h_2,H\!\!>_h<\!\!N_1,N_2\!\!>_r$\;\;if $\{l,l+1,\cdots,l+n\}\cap \mathbf{dom}(h_1\cdot h_2 \cdot H)=\emptyset$ and $h_1'=[h_1 \mid l\mapsto v_0,\cdots,l+n\mapsto v_n]$
\item $<\!\!\mathbf{dispose}\;l\!\!>_k<\!\!s\!\!>_s<\!\!h_1,h_2,H\!\!>_h<\!\!N_1,N_2\!\!>_r\xrightarrow[\Gamma, A, B]{}<\!\!\cdot\!\!>_k<\!\!s\!\!>_s<\!\!h_1\setminus l,h_2,H\!\!>_h<\!\!N_1,N_2\!\!>_r$\\
if $l\in \mathbf{dom}(h_1)$
\item $<\!\!\mathbf{dispose}\;l\!\!>_k<\!\!s\!\!>_s<\!\!h_1,h_2,H\!\!>_h<\!\!N_1,N_2\!\!>_r\xrightarrow[\Gamma, A, B]{}\mathbf{abort}$ \;\;
if $l\notin \mathbf{dom}(h_1)$
\item $<\!\!\mathbf{try}\;r\!\!>_k<\!\!s\!\!>_s<\!\!h_1,h_2,H\!\!>_h<\!\!N_1,N_2\!\!>_r\xrightarrow[\Gamma, A, B]{}<\!\!\cdot\!\!>_k<\!\!s\!\!>_s<\!\!h_1,h_2,H\!\!>_h<\!\!N_1,N_2\!\!>_r$ \\
if $r\in N_1 \cup N_2$
\item $<\!\!\mathbf{acq}\;r\!\!>_k<\!\!s\!\!>_s<\!\!h_1,h_2,H\!\!>_h<\!\!N_1,N_2\!\!>_r\xrightarrow[\Gamma, A, B]{}<\!\!\cdot\!\!>_k<\!\![s\mid r\mapsto0]\!\!>_s<\!\!h_1\cdot h_r,h_2,H-h_r\!\!>_h<\!\!N_1\cup\{r\},N_2\!\!>_r$\;\;if $r \notin N_1\cup N_2$ , $h_r \subseteq H$ and $<\!\!s\!\!>_s<\!\!h_r\!\!>_h<\!\!\{\}\!\!>_r\models \mathbf{inv}(\Gamma\upharpoonright r)$
\item $<\!\!\mathbf{rel}\;r\!\!>_k<\!\!s\!\!>_s<\!\!h_1,h_2,H\!\!>_h<\!\!N_1,N_2\!\!>_r\xrightarrow[\Gamma, A, B]{}<\!\!\cdot\!\!>_k<\!\![s\mid r\mapsto1]\!\!>_s<\!\!h_1- h_r,h_2,H\cup h_r\!\!>_h<\!\!N_1-\{r\},N_2\!\!>_r$\;\;if $r \in N_1$ , $h_r \subseteq h_1$ and $<\!\!s\!\!>_s<\!\!h_r\!\!>_h<\!\!\{\}\!\!>_r\models \mathbf{inv}(\Gamma\upharpoonright r)$
\item $<\!\!\mathbf{rel}\;r\!\!>_k<\!\!s\!\!>_s<\!\!h_1,h_2,H\!\!>_h<\!\!N_1,N_2\!\!>_r\xrightarrow[\Gamma, A, B]{}\mathbf{abort}$ \;\;\;\;\;\;
if $\forall h_r\subseteq h_1 \;\;<\!\!s\!\!>_s<\!\!h_r\!\!>_h<\!\!\{\}\!\!>_r \models \neg\mathbf{inv}(\Gamma\upharpoonright r)$
\end{itemize}
The semantics rules describe the process's execution of action $\lambda$ and its impact on the ownership of heap and resources. The execution of $\lambda$ is legal only if the ownership rule is respected and the separation attribute is maintained. If the execution of $\lambda$ violates the rules, an $\mathbf{abort}$ result occurs. By swapping the roles of process and environment, we gain environment moves $\rightsquigarrow_{\Gamma, A, B}$, which respects $\Gamma$ and does not modify identifiers in $A$ and $B$.
\begin{defi}
$$<\!\!\lambda\!\!>_k<\!\!s\!\!>_s<\!\!h_1,h_2,H\!\!>_h<\!\!N_1,N_2\!\!>_r\rightsquigarrow_{\Gamma, A,B}<\!\!\lambda\!\!>_k<\!\!s'\!\!>_s<\!\!h_1,h_2',H'\!\!>_h<\!\!N_1,N_2'\!\!>_r$$
iff there is an action $\mu$ such that $\mathbf{writes}(\mu)\cap A =\emptyset$ and $\mathbf{writes}(\mu)\cap B =\emptyset$ and
$$<\!\!\mu\!\!>_k<\!\!s\!\!>_s<\!\!h_2,h_1,H\!\!>_h<\!\!N_2,N_1\!\!>_r\xrightarrow[\Gamma, A', B']{}<\!\!\cdot\!\!>_k<\!\!s'\!\!>_s<\!\!h_2',h_1,H'\!\!>_h<\!\!N_2',N_1\!\!>_r$$
We then define
$$<\!\!\lambda\!\!>_k<\!\!s\!\!>_s<\!\!h_1,h_2,H\!\!>_h<\!\!N_1,N_2\!\!>_r\xLongrightarrow[\Gamma, A, B]{}<\!\!\cdot\!\!>_k<\!\!s'\!\!>_s<\!\!h_1',h_2',H'\!\!>_h<\!\!N_1',N_2'\!\!>_r$$
iff
$$<\!\!\lambda\!\!>_k<\!\!s\!\!>_s<\!\!h_1,h_2,H\!\!>_h<\!\!N_1,N_2\!\!>_r\rightsquigarrow_{\Gamma, A' ,B'}^*<\!\!\lambda\!\!>_k<\!\!s''\!\!>_s<\!\!h_1,h_2'',H''\!\!>_h<\!\!N_1,N_2''\!\!>_r\xrightarrow[\Gamma, A, B]{}$$
$$<\!\!\cdot\!\!>_k<\!\!s'''\!\!>_s<\!\!h_1',h_2'',H'''\!\!>_h<\!\!N_1',N_2''\!\!>_r\rightsquigarrow_{\Gamma, A'', B''}^*<\!\!\cdot\!\!>_k<\!\!s'\!\!>_s<\!\!h_1',h_2',H'\!\!>_h<\!\!N_1',N_2'\!\!>_r$$
\end{defi}
\subsection*{Trace}
We call a non-empty finite or infinite sequence of actions as a trace. $\epsilon$ stands for empty trace. Let $i$ is an identifier, $l$ is an address, $r$ is a resource name, $\lambda$ is an action, $\alpha, \beta$ are traces, and $T, T_1, T_2$ are sets of traces. $\alpha\setminus r$ means to replace all the resource actions on $r$ in $\alpha$ with $\delta$. $\alpha\beta$ represents the trace obtained by connecting $\alpha$ and $\beta$. If $\alpha$ is infinite, then $\alpha\beta$ is also infinite. $T_1T_2$ is a set of traces connected by the trace in $T_1$ and the trace in $T_2$. $T^0=\{\delta\}$ and $T^{n+1}=T^nT$. The semantic rules extended to the trace are as follows.
\begin{itemize}
  \item $<\!\!\lambda \alpha\!\!>_k<\!\!s\!\!>_s<\!\!h_1,h_2,H\!\!>_h<\!\!N_1,N_2\!\!>_r\xrightarrow[\Gamma, A,B]{}<\!\!\alpha\!\!>_k<\!\!s'\!\!>_s<\!\!h_1',h_2,H'\!\!>_h<\!\!N_1',N_2\!\!>_r$\\
      if $<\!\!\lambda\!\!>_k<\!\!s\!\!>_s<\!\!h_1,h_2,H\!\!>_h<\!\!N_1,N_2\!\!>_r\xrightarrow[\Gamma, A,B]{}<\!\!\cdot\!\!>_k<\!\!s'\!\!>_s<\!\!h_1',h_2,H'\!\!>_h<\!\!N_1',N_2\!\!>_r$
  \item $<\!\!\lambda \alpha\!\!>_k<\!\!s\!\!>_s<\!\!h_1,h_2,H\!\!>_h<\!\!N_1,N_2\!\!>_r\xLongrightarrow[\Gamma, A, B]{}<\!\!\alpha\!\!>_k<\!\!s'\!\!>_s<\!\!h_1',h_2',H'\!\!>_h<\!\!N_1',N_2'\!\!>_r$\\
      if $<\!\!\lambda\!\!>_k<\!\!s\!\!>_s<\!\!h_1,h_2,H\!\!>_h<\!\!N_1,N_2\!\!>_r\xLongrightarrow[\Gamma, A, B]{}<\!\!\cdot\!\!>_k<\!\!s'\!\!>_s<\!\!h_1',h_2',H'\!\!>_h<\!\!N_1',N_2'\!\!>_r$
\end{itemize}
$T_{[i:1]}$ is the subset of $T$. If $\alpha\in T_{[i:1]}$, then
$$<\!\!\alpha\!\!>_k<\!\![s\mid i\mapsto1]\!\!>_s<\!\!h_1,h_2,H\!\!>_h<\!\!N_1,N_2\!\!>_r\xrightarrow[\Gamma, A,B]{}^*<\!\!\cdot\!\!>_k<\!\!s'\!\!>_s<\!\!h_1',h_2,H'\!\!>_h<\!\!N_1',N_2\!\!>_r$$
\subsection*{Semantics of expressions and commands}
Let $e=i_1+i_2$, iteratively using structure rules and semantics rules, we get
$$i_1+i_2=(i_1=v_1\curvearrowright \square+i_2)\rightharpoonup(v_1+i_2)\rightharpoonup(i_2=v_2\curvearrowright v_1+\square)\rightharpoonup(v_1+v_2)\rightarrow (v_1+_{Int}v_2)$$
If we do not express the structural rules explicitly, $i_1+i_2$ can be simplified to
$$i_1+i_2= (i_1=v_1)(i_2=v_2)\curvearrowright(v_1+_{Int}v_2)$$
where $(i_1=v_1)(i_2=v_2)$ is a trace, $(v_1+_{Int}v_2)$ is an integer value and $\curvearrowright$ a delimiter to split the trace and the integer value. In fact, K uses strictness attribute in order to avoid writing obvious structural rules.
Then, the semantic of expression $e_1+e_2$ is expressed as a trace paired with an integer value. Since the expression is pure, the only action involved in such a trace is read action ($i=v$).\\
Let $v, v_1, v_2, \ldots, v_n$ are integer values, $\alpha, \alpha_1, \alpha_2, \alpha_3, \rho, \rho_1, \rho_2$ are traces, $l$ is address, $\lambda, \mu$ are actions, $e, e_1, e_2, \ldots, e_n$ are integer expressions, $b$ is boolean expression, $i$ is identifier,  $E\!=\![e_1,e_2,\!\ldots\!, e_n]$ is list expression, $V=[v_1,v_2,\ldots, v_n]$ is list integer value and $[\![]\!]$ is the semantic function which is given by structural induction as follows.
\begin{align*}
  [\![6]\!]&=\delta\curvearrowright6\\
  [\![i]\!]&=i=v\curvearrowright v\\
  [\![e_1+e_2]\!]&=\rho_1\rho_2\curvearrowright (v_1+_{Int}v_2)\;\;\;\; where\; [\![e_1]\!]=\rho_1\curvearrowright v_1\; and\; [\![e_2]\!]=\rho_2\curvearrowright v_2\\
  [\![(e_1,e_2,\ldots,e_i,\ldots,e_n)]\!]&=\rho_1\ldots\rho_i\ldots\rho_n\curvearrowright[v_1,v_2,\ldots,v_i,\ldots,v_n]\;\;\;\;where\;[\![e_i]\!]=\rho_i\curvearrowright v_i
\end{align*}
Similarly,
\begin{align*}
  [\![true]\!]&=\delta\curvearrowright true\\
  [\![false]\!]&=\delta\curvearrowright false\\
  [\![b]\!]_{true}&=\rho_1\curvearrowright true\\
  [\![b]\!]_{false}&=\rho_2\curvearrowright false
\end{align*}
The command is relatively complicated and can be expressed as a set of traces, which can be either finite or infinite.
\begin{align*}
  [\![\mathbf{skip}]\!]&=\{\delta\}\\
  [\![i:=e]\!]&=\{\rho(i:=v)\mid[\![e]\!]=\rho\curvearrowright v\}\\
  [\![i:=[e]]\!]&=\{\rho([l]=v)(i:=v)\mid [\![e]\!]=\rho\curvearrowright l\}\\
  [\![[e_1]:=e_2]\!]&=\{\rho_1\rho_2([l]:=v)\mid [\![e_1]\!]=\rho_1\curvearrowright l\;and\;[\![e_2]\!]=\rho_2\curvearrowright v\}\\
  [\![i:=\mathbf{cons}\;E]\!]&=\{\rho(\mathbf{alloc}(l,V))(i:=l)\mid \rho\curvearrowright V\} \\
  [\![\mathbf{dispose}\;e]\!]&=\{\rho(\mathbf{dispose}\;l)\mid \rho\curvearrowright l\}\\
  [\![k_1;k_2]\!]&=[\![k_1]\!][\![k_2]\!]\\
  [\![\mathbf{if}\;b\;\mathbf{then}\;k_1\;\mathbf{else}\;k_2]\!]&=[\![b]\!]_{true}[\![k_1]\!]\cup[\![b]\!]_{false}[\![k_2]\!]\\
  [\![\mathbf{while}\;b\;\mathbf{do}\;k]\!]&=[\![\mathbf{if}\;b\;\mathbf{then}\;(k;\mathbf{while}\;b\;\mathbf{do}\;k)\;\mathbf{else}\;\mathbf{skip}]\!]\\
  [\![\mathbf{resource}\;r\;\mathbf{in}\;k]\!]&=\{\alpha\backslash r\mid \alpha\in [\![k]\!]_{[r:1]}\}\\
  [\![\mathbf{with}\;r\;\mathbf{when}\;b\;\mathbf{do}\;k]\!]&=wait^*enter\cup wait^\omega\\
  where&\;wait=\{\mathbf{try}\;r\}\cup\{(\mathbf{acq}\;r)\rho(\mathbf{rel}\;r)\mid \rho\in [\![b]\!]_{false}\}\\
  and  &\;enter=\{(\mathbf{acq}\;r)\rho\alpha(\mathbf{rel}\;r)\mid [\![b]\!]_{true}\; and\; \alpha\in[\![k]\!]\}\\
  [\![k_1\|k_2]\!]&=\cup\{\alpha_{1\{A_1,B_1\}}\|_{\{A_2,B_2\}}\alpha_2\mid \alpha_1\in [\![k_1]\!]\;and\; \alpha_2\in [\![k_2]\!]\\
  and & \;\mathbf{writes}(\alpha_1)\cap A_2 =\emptyset\;and\;\mathbf{writes}(\alpha_1)\cap B_2 =\emptyset\\
  and & \;\mathbf{writes}(\alpha_2)\cap A_1 =\emptyset\;and\;\mathbf{writes}(\alpha_2)\cap B_1 =\emptyset\}
\end{align*}
where
$$\lambda\alpha_{1\{A_1,B_1\}}\|_{\{A_2,B_2\}}\mu\alpha_2=\{\lambda\alpha_3\mid \alpha_3\in \alpha_{1\{A_1,B_1\}}\|_{\{A_2,B_2\}}\mu\alpha_2\}$$
$$\cup\{\mu\alpha_3\mid \alpha_3\in \lambda\alpha_{1\{A_1,B_1\}}\|_{\{A_2,B_2\}}\alpha_2\}$$
$$\cup\{\mathbf{abort}\mid \mathbf{writes}(\lambda)\cap \mathbf{free}(\mu)\neq\emptyset\;
or\;\mathbf{writes}(\mu)\cap \mathbf{free}(\lambda)\neq\emptyset\}$$
\subsection*{Validity}
\begin{defi}
The well-formed assertion
$\Gamma \vdash_{A,B} \exists X((o=<\!\!k\!\!>_k<\!\!s\!\!>_s<\!\!h_1,h_2,H\!\!>_h<\!\!N_1,N_2\!\!>_r)\land p)\Downarrow \exists X'((o=<\!\!\cdot\!\!>_k<\!\!s'\!\!>_s<\!\!h_1',h_2',H'\!\!>_h<\!\!N_1',N_2'\!\!>_r)\land q)$ is valid iff  $<\!\!s_c\!\!>_s<\!\!h_{1c},h_{2c},H_c\!\!>_h<\!\!N_{1c},N_{2c}\!\!>_r\in \Sigma_\Gamma$ with $\mathbf{owned}(\Gamma)\cup A \subseteq \mathbf{dom}(s_c)$, if there exits a map $\tau:\mathit{Var}\to Int$ such that $(<\!\!k\!\!>_k<\!\!s_c\!\!>_s<\!\!h_{1c},h_{2c},H_c\!\!>_h<\!\!N_{1c},N_{2c}\!\!>_r,\tau)\models_B \exists X((o=<\!\!k\!\!>_k<\!\!s\!\!>_s<\!\!h_1,h_2,H\!\!>_h<\!\!N_1,N_2\!\!>_r)\land p)$, then for all traces $\alpha\in [\![k]\!]$,
\begin{itemize}
  \item $\neg(<\!\!\alpha\!\!>_k<\!\!s_c\!\!>_s<\!\!h_{1c},h_{2c},H_c\!\!>_h<\!\!N_{1c},N_{2c}\!\!>_r\xLongrightarrow[\Gamma, A, B]{}^*\mathbf{abort})$
  \item if $<\!\!\alpha\!\!>_k<\!\!s_c\!\!>_s<\!\!h_{1c},h_{2c},H_c\!\!>_h<\!\!N_{1c},N_{2c}\!\!>_r\xLongrightarrow[\Gamma, A, B]{}^*<\!\!\cdot\!\!>_k<\!\!s_c'\!\!>_s<\!\!h_{1c}',h_{2c}',H_c'\!\!>_h<\!\!N_{1c}',N_{2c}'\!\!>_r$, then $(<\!\!\cdot\!\!>_k<\!\!s_c'\!\!>_s<\!\!h_{1c}',h_{2c}',H_c'\!\!>_h<\!\!N_{1c}',N_{2c}'\!\!>_r,\tau)\models_B \exists X'((o=<\!\!\cdot\!\!>_k<\!\!s'\!\!>_s<\!\!h_1',h_2',H'\!\!>_h<\!\!N_1',N_2'\!\!>_r)\land q)$
\end{itemize}
\end{defi}
\section{Soundness}
The proof of soundness of the inference rules is a rule-by-rule case analysis. We start with some important properties of the semantics.
\begin{lem}
Let $h\bot h_3$, $h_1\bot h_2$, $h=h_1\cdot h_2$, $N \cap N_3=\emptyset$, $N=N_1 \cup N_2$, $N_1 \cap N_2=\emptyset$, $\lambda$ is an action.
\begin{enumerate}
\item If $<\!\!\lambda\!\!>_k<\!\!s\!\!>_s<\!\!h,h_3,H\!\!>_h<\!\!N,N_3\!\!>_r\xLongrightarrow[\Gamma, A_1\cup A_2, B_1\cup B_2]{}^*\mathbf{abort}$, then $<\!\!\lambda\!\!>_k<\!\!s\!\!>_s<\!\!h_1,h_3\cdot h_2,H\!\!>_h<\!\!N_1,N_3\cup N_2\!\!>_r\xLongrightarrow[\Gamma, A_1,B_1]{}^*\mathbf{abort}$
\item If $<\!\!\lambda\!\!>_k<\!\!s\!\!>_s<\!\!h,h_3,H\!\!>_h<\!\!N,N_3\!\!>_r\xLongrightarrow[\Gamma, A_1\cup A_2, B_1\cup B_2]{}^*<\!\!\cdot\!\!>_k<\!\!s'\!\!>_s<\!\!h',h_3',H'\!\!>_h<\!\!N',N_3'\!\!>_r$, then $<\!\!\lambda\!\!>_k<\!\!s\!\!>_s<\!\!h_1,h_3\cdot h_2,H\!\!>_h<\!\!N_1,N_3\cup N_2\!\!>_r\xLongrightarrow[\Gamma, A_1, B_1]{}^*\mathbf{abort}$ or there are $h_1'$, $N_1'$ such that $h'=h_1'\cdot h_2$ and $N_1'\cap N_2=\emptyset$ and $N'=N_1'\cup N_2$ and $<\!\!\lambda\!\!>_k<\!\!s\!\!>_s<\!\!h_1,h_3\cdot h_2,H\!\!>_h<\!\!N_1,N_3\cup N_2\!\!>_r\xLongrightarrow[\Gamma, A_1, B_1]{}^*<\!\!\cdot\!\!>_k<\!\!s'\!\!>_s<\!\!h_1',h_3'\cdot h_2,H'\!\!>_h<\!\!N_1',N_3'\cup N_2\!\!>_r$.
\end{enumerate}
\end{lem}
The proof is shown in Appendix A (Lemma 4.1).
\begin{lem}
$h\bot h_3$, $h_1\bot h_2$, $h=h_1\cdot h_2$, $N \cap N_3=\emptyset$, $N=N_1 \cup N_2$, $N_1 \cap N_2=\emptyset$, $\alpha$ is an trace.
\begin{enumerate}
\item If $<\!\!\alpha\!\!>_k<\!\!s\!\!>_s<\!\!h,h_3,H\!\!>_h<\!\!N,N_3\!\!>_r\xLongrightarrow[\Gamma, A_1\cup A_2, B_1\cup B_2]{}^*\mathbf{abort}$, then $<\!\!\alpha\!\!>_k<\!\!s\!\!>_s<\!\!h_1,h_3\cdot h_2,H\!\!>_h<\!\!N_1,N_3\cup N_2\!\!>_r\xLongrightarrow[\Gamma, A_1,B_1]{}^*\mathbf{abort}$
\item If $<\!\!\alpha\!\!>_k<\!\!s\!\!>_s<\!\!h,h_3,H\!\!>_h<\!\!N,N_3\!\!>_r\xLongrightarrow[\Gamma, A_1\cup A_2, B_1\cup B_2]{}^*<\!\!\cdot\!\!>_k<\!\!s'\!\!>_s<\!\!h',h_3',H'\!\!>_h<\!\!N',N_3'\!\!>_r$, then $<\!\!\alpha\!\!>_k<\!\!s\!\!>_s<\!\!h_1,h_3\cdot h_2,H\!\!>_h<\!\!N_1,N_3\cup N_2\!\!>_r\xLongrightarrow[\Gamma, A_1, B_1]{}^*\mathbf{abort}$ or there are $h_1'$, $N_1'$ such that $h'=h_1'\cdot h_2$ and $N_1'\cap N_2=\emptyset$ and $N'=N_1'\cup N_2$ and $<\!\!\alpha\!\!>_k<\!\!s\!\!>_s<\!\!h_1,h_3\cdot h_2,H\!\!>_h<\!\!N_1,N_3\cup N_2\!\!>_r\xLongrightarrow[\Gamma, A_1, B_1]{}^*<\!\!\cdot\!\!>_k<\!\!s'\!\!>_s<\!\!h_1',h_3'\cdot h_2,H'\!\!>_h<\!\!N_1',N_3'\cup N_2\!\!>_r$.
\end{enumerate}
\end{lem}
Lemma 4.2 is a generalization of Lemma 4.1, the proof of Lemma 4.2 is an obvious induction proof.
\begin{thm}
Let $\mathbf{mod}(\alpha_1)\cap A_2=\mathbf{mod}(\alpha_2)\cap A_1=\emptyset$, $\mathbf{free}(\alpha_1)\subseteq \mathbf{owned}(\Gamma)\cup A_1$, $\mathbf{free}(\alpha_2)\subseteq \mathbf{owned}(\Gamma)\cup A_2$, $N_1 \cap N_2=\emptyset$, $N=N_1 \cup N_2$, $N \cap N_3=\emptyset$, $h=h_1\cdot h_2$, $h\bot h_3$, and $\alpha\in \alpha_{1\{A_1,B_1\}}\|_{\{A_2,B_2\}}\alpha_2$
\begin{enumerate}
  \item If $<\!\!\alpha\!\!>_k<\!\!s\!\!>_s<\!\!h,h_3,H\!\!>_h<\!\!N,N_3\!\!>_r\xLongrightarrow[\Gamma, A_1\cup A_2, B_1\cup B_2]{}^*\mathbf{abort}$, then $<\!\!\alpha_1\!\!>_k<\!\!s\!\!>_s<\!\!h_1,h_3\cdot h_2,H\!\!>_h<\!\!N_1,N_3\cup N_2\!\!>_r\xLongrightarrow[\Gamma, A_1,B_1]{}^*\mathbf{abort}$ or $<\!\!\alpha_2\!\!>_k<\!\!s\!\!>_s<\!\!h_2,h_3\cdot h_1,H\!\!>_h<\!\!N_2,N_3\cup N_1\!\!>_r\xLongrightarrow[\Gamma, A_2, B_2]{}^*\mathbf{abort}$
  \item If $<\!\!\alpha\!\!>_k<\!\!s\!\!>_s<\!\!h,h_3,H\!\!>_h<\!\!N,N_3\!\!>_r\xLongrightarrow[\Gamma, A_1\cup A_2, B_1\cup B_2]{}^*<\!\!\cdot\!\!>_k<\!\!s'\!\!>_s<\!\!h',h_3',H'\!\!>_h<\!\!N',N_3'\!\!>_r$, then $<\!\!\alpha_1\!\!>_k<\!\!s\!\!>_s<\!\!h_1,h_3\cdot h_2,H\!\!>_h<\!\!N_1,N_3\cup N_2\!\!>_r\xLongrightarrow[\Gamma, A_1, B_1]{}^*\mathbf{abort}$, $<\!\!\alpha_2\!\!>_k<\!\!s\!\!>_s<\!\!h_2,h_3\cdot h_1,H\!\!>_h<\!\!N_2,N_3\cup N_1\!\!>_r\xLongrightarrow[\Gamma, A_2, B_2]{}^*\mathbf{abort}$ or there are $h_1'$, $h_2'$, $N_1'$, $N_2'$ such that $N_1'\cap N_2'=\emptyset$, $N'=N_1'\cup N_2'$, $h'=h_1'\cdot h_2'$ and
      $<\!\!\alpha_1\!\!>_k<\!\!s\!\!>_s<\!\!h_1,h_3\cdot h_2,H\!\!>_h<\!\!N_1,N_3\cup N_2\!\!>_r\xLongrightarrow[\Gamma, A_1, B_1]{}^*<\!\!\cdot\!\!>_k<\!\!s'\!\!>_s<\!\!h_1',h_3'\cdot h_2',H'\!\!>_h<\!\!N_1',N_3'\cup N_2'\!\!>_r$, $<\!\!\alpha_2\!\!>_k<\!\!s\!\!>_s<\!\!h_2,h_3\cdot h_1,H\!\!>_h<\!\!N_2,N_3\cup N_1\!\!>_r\xLongrightarrow[\Gamma, A_2, B_2]{}^*<\!\!\cdot\!\!>_k<\!\!s'\!\!>_s<\!\!h_2',h_3'\cdot h_1',H'\!\!>_h<\!\!N_2',N_3'\cup N_1'\!\!>_r$
\end{enumerate}
\end{thm}
The proof is shown in Appendix A (Theorem 4.3).
\begin{thm}
\begin{enumerate}
  \item If $<\!\!s\!\!>_s<\!\!h_1\cdot h,h_2,H\!\!>_h<\!\!N_1,N_2\!\!>_r\in \Sigma_\Gamma$ and $<\!\!s\!\!>_s<\!\!h\!\!>_h<\!\!\{\}\!\!>_r\models R$ and $free(R)\subseteq X$ and $r\notin \mathbf{dom}(\Gamma)$, then $<\!\![s\mid r\mapsto 1]\!\!>_s<\!\!h_1,h_2,H\cdot h\!\!>_h<\!\!N_1,N_2\!\!>_r\in \Sigma_{\Gamma,r(X):R}$;
  \item If $r\notin \mathbf{dom}(\Gamma)$ and $\beta\in [\![k]\!]_{r:1}$ and $<\!\!\beta\backslash r\!\!>_k<\!\!s\!\!>_s<\!\!h_1\cdot h,h_2,H\!\!>_h<\!\!N_1,N_2\!\!>_r\xLongrightarrow[\Gamma, A, B]{}^*\mathbf{abort}$, then $<\!\!\beta\!\!>_k<\!\![s\mid r\mapsto 1]\!\!>_s<\!\!h_1,h_2,H\cdot h\!\!>_h<\!\!N_1,N_2\!\!>_r\xLongrightarrow[(\Gamma,r(X):R), A, B]{}^*\mathbf{abort}$;
  \item If $r\notin \mathbf{dom}(\Gamma)$ and $\beta\in [\![k]\!]_{r:1}$ and $<\!\!\beta\backslash r\!\!>_k<\!\!s\!\!>_s<\!\!h_1\cdot h,h_2,H\!\!>_h<\!\!N_1,N_2\!\!>_r\xLongrightarrow[\Gamma, A, B]{}^*<\!\!\cdot\!\!>_k<\!\!s'\!\!>_s<\!\!h_1'\cdot h',h_2',H'\!\!>_h<\!\!N_1',N_2'\!\!>_r$, then either  $<\!\!\beta\!\!>_k<\!\![s\mid r\mapsto 1]\!\!>_s<\!\!h_1,h_2,H\cdot h\!\!>_h<\!\!N_1,N_2\!\!>_r\xLongrightarrow[(\Gamma,r(X):R), A, B]{}^*\mathbf{abort}$ or $<\!\!\beta\!\!>_k<\!\![s\mid r\mapsto 1]\!\!>_s<\!\!h_1,h_2,H\cdot h\!\!>_h<\!\!N_1,N_2\!\!>_r\xLongrightarrow[(\Gamma,r(X):R), A, B]{}^*<\!\!\cdot\!\!>_k<\!\![s'\mid r\mapsto 1]\!\!>_s<\!\!h_1',h_2',H'\cdot h'\!\!>_h<\!\!N_1',N_2'\!\!>_r$.
\end{enumerate}
\end{thm}
The Theorem 4.4 is very intuitive, we omit the proof.
\begin{thm}(Soundness)\\
Every provable assertion $\Gamma \vdash_{A,B} \exists X((o=<\!\!k\!\!>_k<\!\!s\!\!>_s<\!\!h_1,h_2,H\!\!>_h<\!\!N_1,N_2\!\!>_r)\land p)\Downarrow \exists X'((o=<\!\!\cdot\!\!>_k<\!\!s'\!\!>_s<\!\!h_1',h_2',H'\!\!>_h<\!\!N_1',N_2'\!\!>_r)\land q)$ is valid.
\end{thm}
\begin{proof}
The proof is by induction on the length of the derivation, that every provable assertion is valid. For some of the rules this is fairly easy, we omit the proof. We only provide proof details for PARALLEL and RESOURCE.
\begin{itemize}
  \item PARALLEL; Let $LS=<\!\!s\!\!>_s<\!\!h_1\cdot h_2, h_3,H\!\!>_h<\!\!N_1\cup N_2, N_3\!\!>_r$ and $LS_1=<\!\!s\!\!>_s<\!\!h_1,h_2\cdot h_3,H\!\!>_h<\!\!N_1,N_2\cup N_3\!\!>_r$ and $LS_1'=<\!\!s'\!\!>_s<\!\!h_1',h_2'\cdot h_3',H'\!\!>_h<\!\!N_1',N_2'\cup N_3'\!\!>_r$ and $LS_2=<\!\!s\!\!>_s<\!\!h_2,h_1\cdot h_3,H\!\!>_h<\!\!N_2,N_1\cup N_3\!\!>_r$ and $LS_2'=<\!\!s'\!\!>_s<\!\!h_2',h_1'\cdot h_3',H'\!\!>_h<\!\!N_2',N_1'\cup N_3'\!\!>_r$ and
      $LS_f=<\!\!s'\!\!>_s<\!\!h_1'\cdot h_2',h_3',H'\!\!>_h<\!\!N_1'\cup N_2',N_3'\!\!>_r$.\\
      Suppose
      $\Gamma \vdash_{A_1,B_1} \exists X_1(o=<\!\!k_1\!\!>_k LS_1\land p_1)\Downarrow\exists X_2(o=<\!\!\cdot\!\!>_k LS_1'\land q_1)$ and $\Gamma \vdash_{A_2,B_2} \exists X_1(o=<\!\!k_2\!\!>_k LS_2\land p_2)\Downarrow\exists X_2(o=<\!\!\cdot\!\!>_k LS_2'\land q_2)$ are well-formed and valid, and $\mathbf{mod}(k_1) \cap A_2=\mathbf{mod}(k_1) \cap B_2=\mathbf{mod}(k_2) \cap A_1 =\mathbf{mod}(k_2) \cap B_1 = \emptyset$ and $N_1\cap N_2=\emptyset$. We next show that $\Gamma \vdash_{A_1\cup A_2,B_1\cup B_2} \exists X_1 (o=<\!\!k_1 \parallel k_2\!\!>_k LS\land (p_1*p_2))\Downarrow\exists X_2 (o=<\!\!\cdot\!\!>_k LS_f\land (q_1* q_2))$ is valid.\\
      Let $\tau:\mathit{Var}\to Int$ and $<\!\!s_c\!\!>_s<\!\!h_{c},h_{3c},H_c\!\!>_h<\!\!N_{c},N_{3c}\!\!>_r\in \Sigma_\Gamma$ and $(<\!\!k_1 \parallel k_2\!\!>_k<\!\!s_c\!\!>_s<\!\!h_{c},h_{3c},H_c\!\!>_h<\!\!N_{c},N_{3c}\!\!>_r,\tau)\models_{B_1\cup B_2} \exists X_1 (o=<\!\!k_1 \parallel k_2\!\!>_k LS\land (p_1*p_2))$, then we have
       \begin{enumerate}
         \item there exists some $\theta_\tau:\mathit{Var}\to Int$ with $\theta_\tau\!\!\upharpoonright_{\mathit{Var}/X_1}=\tau\!\!\upharpoonright_{\mathit{Var}/X_1}$;
         \item $s_c\upharpoonright_{B_1\cup B_2}=\theta_\tau(s)\upharpoonright_{B_1\cup B_2}$;
         \item $h_{c}=h_1\cdot h_2$ and $N_{c} = N_1\cup N_2$;
         \item $<\!\!s_c\!\!>_s<\!\!h_{c}\!\!>_h<\!\!\{\}\!\!>_r\models p_1\ast p_2$.
       \end{enumerate}
       Let $h_{1c}\bot h_{2c}$ and $h_c=h_{1c}\cdot h_{2c}$ and $<\!\!s_c\!\!>_s<\!\!h_{1c}\!\!>_h<\!\!\{\}\!\!>_r\models p_1$ and $<\!\!s_c\!\!>_s<\!\!h_{2c}\!\!>_h<\!\!\{\}\!\!>_r\models p_2$ and $N_{1c}=N_1$ and $N_{2c}=N_2$ and $N=N_1\cup N_2$.
       Since $\alpha_1\in [\![k_1]\!]$ and $\alpha_2\in [\![k_2]\!]$ and $\mathbf{mod}(k_1) \cap A_2=\mathbf{mod}(k_2) \cap A_1 =\emptyset$, then $\mathbf{mod}(\alpha_1)\cap A_2=\mathbf{mod}(\alpha_2)\cap A_1=\emptyset$.\\
       Let $\alpha\in \alpha_{1\{A_1,B_1\}}\|_{\{A_2,B_2\}}\alpha_2$,
       \begin{enumerate}
         \item if $<\!\!\alpha\!\!>_k<\!\!s_c\!\!>_s<\!\!h_{c},h_{3c},H_c\!\!>_h<\!\!N_{c},N_{3c}\!\!>_r\xLongrightarrow[\Gamma, A_1\cup A_2, B_1\cup B_2]{}^*\mathbf{abort}$, then by Theorem 4.3, either $<\!\!\alpha_1\!\!>_k<\!\!s_c\!\!>_s<\!\!h_{1c},h_{3c}\cdot h_{2c}, H_c\!\!>_h<\!\!N_{1c},N_{3c}\cup N_{2c}\!\!>_r\xLongrightarrow[\Gamma, A_1, B_1]{}^*\mathbf{abort}$ or $<\!\!\alpha_2\!\!>_k<\!\!s_c\!\!>_s<\!\!h_{2c},h_{3c}\cdot h_{1c},H_c\!\!>_h<\!\!N_{2c},N_{3c}\cup N_{1c}\!\!>_r\xLongrightarrow[\Gamma, A_2, B_2]{}^*\mathbf{abort}$. Neither case is possible because they contradict the validity of the hypothesis $\Gamma \vdash_{A_1,B_1} \exists X_1(o=<\!\!k_1\!\!>_k LS_1\land p_1)\Downarrow\exists X_2(o=<\!\!\cdot\!\!>_k LS_1'\land q_1)$ and $\Gamma \vdash_{A_2,B_2} \exists X_1(o=<\!\!k_2\!\!>_k LS_2\land p_2)\Downarrow\exists X_2(o=<\!\!\cdot\!\!>_k LS_2'\land q_2)$.
         \item if $<\!\!\alpha\!\!>_k<\!\!s_c\!\!>_s<\!\!h_{c},h_{3c},H_c\!\!>_h<\!\!N_{c},N_{3c}\!\!>_r\xLongrightarrow[\Gamma, A_1\cup A_2, B_1\cup B_2]{}^*<\!\!\cdot\!\!>_k<\!\!s_c'\!\!>_s<\!\!h_{c}',h_{3c}',H_c'\!\!>_h<\!\!N_{c}',N_{3c}'\!\!>_r$, then by Theorem 4.3, there are $h_{1c}'$, $h_{2c}'$, $N_{1c}'$, $N_{2c}'$ such that $N_{1c}'\cap N_{2c}'=\emptyset$, $N'=N_{1c}'\cup N_{2c}'$, $h'=h_{1c}'\cdot h_{2c}'$ and
      $<\!\!\alpha_1\!\!>_k<\!\!s_c\!\!>_s<\!\!h_{1c},h_{3c}\cdot h_{2c},H_c\!\!>_h<\!\!N_{1c},N_{3c}\cup N_{2c}\!\!>_r\xLongrightarrow[\Gamma, A_1, B_1]{}^*<\!\!\cdot\!\!>_k<\!\!s'\!\!>_s<\!\!h_{1c}',h_{3c}'\cdot h_{2c}',H_c'\!\!>_h<\!\!N_{1c}',N_{3c}'\cup N_{2c}'\!\!>_r$ and $<\!\!\alpha_2\!\!>_k<\!\!s\!\!>_s<\!\!h_{2c},h_{3c}\cdot h_{1c},H_c\!\!>_h<\!\!N_{2c},N_{3c}\cup N_{1c}\!\!>_r\xLongrightarrow[\Gamma, A_2, B_2]{}^*<\!\!\cdot\!\!>_k<\!\!s'\!\!>_s<\!\!h_{2c}',h_{3c}'\cdot h_{1c}',H_c'\!\!>_h<\!\!N_{2c}',N_{3c}'\cup N_{1c}'\!\!>_r$.\\
      By the induction hypothesis for $\Gamma \vdash_{A_1,B_1} \exists X_1(o=<\!\!k_1\!\!>_k LS_1\land p_1)\Downarrow\exists X_2(o=<\!\!\cdot\!\!>_k LS_1'\land q_1)$, we have $(<\!\!\cdot\!\!>_k<\!\!s_c'\!\!>_s<\!\!h_{1c}',h_{3c}'\cdot h_{2c}' ,H_c'\!\!>_h<\!\!N_{1c}',N_{3c}'\cup N_{2c}'\!\!>_r,\tau)\models_{B_1} \exists X_2(o=<\!\!\cdot\!\!>_k LS_1'\land q_1)$.\\
      By the induction hypothesis for $\Gamma \vdash_{A_2,B_2} \exists X_1(o=<\!\!k_2\!\!>_k LS_2\land p_2)\Downarrow\exists X_2(o=<\!\!\cdot\!\!>_k LS_2'\land q_2)$, we have $(<\!\!\cdot\!\!>_k<\!\!s_c'\!\!>_s<\!\!h_{2c}',h_{3c}'\cdot h_{1c}' ,H_c'\!\!>_h<\!\!N_{2c}',N_{3c}'\cup N_{1c}'\!\!>_r,\tau)\models_{B_2} \exists X_2(o=<\!\!\cdot\!\!>_k LS_2'\land q_2)$.\\
      Hence, $(<\!\!\cdot\!\!>_k<\!\!s_c'\!\!>_s<\!\!h_{c}',h_{3c}',H_c'\!\!>_h<\!\!N_{c}',N_{3c}'\!\!>_r,\tau)\models_{B_1,B_2}\exists X_2 (o=<\!\!\cdot\!\!>_k LS_f\land (q_1* q_2))$.
       \end{enumerate}
  \item  RESOURCE; Let $LS_1=<\!\!s\!\!>_s<\!\!h_1\cdot h,h_2,H\!\!>_h<\!\!N_1,N_2\!\!>_r$, $LS_2=<\!\!s'\!\!>_s<\!\!h_1'\cdot h',h_2',H'\!\!>_h<\!\!N_1',N_2'\!\!>_r$, $LS_1'=<\!\!s\!\!>_s<\!\!h_1,h_2,H\cdot h\!\!>_h<\!\!N_1,N_2\!\!>_r$, $LS_2'=<\!\!s'\!\!>_s<\!\!h_1',h_2',H'\cdot h'\!\!>_h<\!\!N_1',N_2'\!\!>_r$.\\
      Suppose $\Gamma,r(Y):R \vdash_{A,B} \exists X(o=<\!\!k\!\!>_k LS_1'\land p )\Downarrow \exists X'(o=<\!\!\cdot\!\!>_k LS_2'\land q)$ is well-formed and valid.\\
       Let $\tau:\mathit{Var}\to Int$ and $<\!\!s_c\!\!>_s<\!\!h_{1c}\cdot h_c,h_{2c},H_c\!\!>_h<\!\!N_{1c},N_{2c}\!\!>_r\in \Sigma_\Gamma$ and $<\!\!s_c\!\!>_s<\!\!h_{1c}\!\!>_h<\!\!\{\}\!\!>_r\models p$ and $<\!\!s_c\!\!>_s<\!\!h_{c}\!\!>_h<\!\!\{\}\!\!>_r\models R$, then
       $(<\!\!\mathbf{resource}\;r\;\mathbf{in}\;k\!\!>_k<\!\!s_c\!\!>_s<\!\!h_{1c}\cdot h_c,h_{2c},H_c\!\!>_h<\!\!N_{1c},N_{2c}\!\!>_r,\tau)\models_{B} \exists X(o=<\!\!\mathbf{resource}\;r\;\mathbf{in}\;k\!\!>_k LS_1\land (p*R))$.
       Let $\alpha\in [\![k]\!]_{[r:1]}$.
       \begin{enumerate}
         \item If $<\!\!\alpha\backslash r\!\!>_k<\!\!s_c\!\!>_s<\!\!h_{1c}\cdot h_c,h_{2c},H_c\!\!>_h<\!\!N_{1c},N_{2c}\!\!>_r\xLongrightarrow[\Gamma, A\cup Y, B]{}^*\mathbf{abort}$, then, by Theorem 4.4, $<\!\!\alpha\!\!>_k<\!\![s_c\mid r\mapsto 1]\!\!>_s<\!\!h_{1c},h_{2c},H_c\cdot h_c\!\!>_h<\!\!N_{1c},N_{2c}\!\!>_r\xLongrightarrow[(\Gamma,r(Y):R), A, B]{}^*\mathbf{abort}$. However, this is
             possible because they contradict the validity of the hypothesis $\Gamma,r(Y):R \vdash_{A,B} \exists X(o=<\!\!k\!\!>_k LS_1'\land p )\Downarrow \exists X'(o=<\!\!\cdot\!\!>_k LS_2'\land q)$.
         \item If $<\!\!\alpha\backslash r\!\!>_k<\!\!s_c\!\!>_s<\!\!h_{1c}\cdot h_c,h_{2c},H_c\!\!>_h<\!\!N_{1c},N_{2c}\!\!>_r\xLongrightarrow[\Gamma, A\cup Y, B]{}^*<\!\!\cdot\!\!>_k<\!\!s_c'\!\!>_s<\!\!h_{1c}'\cdot h_c',h_{2c}',H_c'\!\!>_h<\!\!N_{1c}',N_{2c}'\!\!>_r$, then by Theorem 4.4,\\
              a), either $<\!\!\alpha\!\!>_k<\!\![s_c\mid r\mapsto 1]\!\!>_s<\!\!h_{1c},h_{2c},H_c\cdot h_c\!\!>_h<\!\!N_{1c},N_{2c}\!\!>_r\xLongrightarrow[(\Gamma,r(Y):R), A, B]{}^*\mathbf{abort}$. This is
             possible as above.\\
              b), or $<\!\!\alpha\!\!>_k<\!\![s_c\mid r\mapsto 1]\!\!>_s<\!\!h_{1c},h_{2c},H_c\cdot h_c\!\!>_h<\!\!N_{1c},N_{2c}\!\!>_r\xLongrightarrow[(\Gamma,r(Y):R), A, B]{}^*<\!\!\cdot\!\!>_k<\!\![s_c'\mid r\mapsto 1]\!\!>_s<\!\!h_{1c}',h_{2c}',H_c'\cdot h_c'\!\!>_h<\!\!N_{1c}',N_{2c}'\!\!>_r$.\\
              By the induction hypothesis for $\Gamma,r(Y):R \vdash_{A,B} \exists X(o=<\!\!k\!\!>_k LS_1'\land p )\Downarrow \exists X'(o=<\!\!\cdot\!\!>_k LS_2'\land q)$, we have $(<\!\!\cdot\!\!>_k<\!\![s_c'\mid r\mapsto 1]\!\!>_s<\!\!h_{1c}',h_{2c}',H_c'\cdot h_c'\!\!>_h<\!\!N_{1c}',N_{2c}'\!\!>_r,\tau)\models_{B} \exists X'(o=<\!\!\cdot\!\!>_k LS_2'\land q)$.\\
              Hence, $(<\!\!\cdot\!\!>_k<\!\!s_c'\!\!>_s<\!\!h_{1c}'\cdot h_c',h_{2c}',H_c'\!\!>_h<\!\!N_{1c}',N_{2c}'\!\!>_r,\tau)\models_{B} \exists X'(o=<\!\!\cdot\!\!>_k LS_2\land (q\ast R))$.
       \end{enumerate}
\end{itemize}
\end{proof}
\section{Relation to CSL}
Before we discuss the relationship between CML and CSL, let's take a closer look at the assertions of CML. The assertion of CML has the form of $\Gamma \vdash_{A,B} \exists X((o=<\!\!k\!\!>_k<\!\!s\!\!>_s<\!\!h_1,h_2,H\!\!>_h<\!\!N_1,N_2\!\!>_r)\land p)\Downarrow \exists X'((o=<\!\!\cdot\!\!>_k<\!\!s'\!\!>_s<\!\!h_1',h_2',H'\!\!>_h<\!\!N_1',N_2'\!\!>_r)\land q)$.
Let's consider a special case where $A=\emptyset$ and $B=\mathbf{dom}(s)$ and $h_1=h_2=H=h_1'=h_2'=H'=\mathbf{emp}$ and $N_1=N_2=N_1'=N_2'=\emptyset$ and $\Gamma=\emptyset$ and $p,q$ do not contain separating conjunction, but is only a first-order logical formula. In this case, CML simplifies to matching logic. \\
The assertion of CSL has the form of $\Gamma \vdash_{A}\{p\}k\{q\}$. Note that the CML pattern provides
more information specifications than CSL. We fix a finite set of program identifiers $Z$ which is large
enough. We assume that $Z_{v},Z_{v}'$ are two
sets of variables called ``semantic clone" of $Z$. Let $s_z, s_z'$ map each program identifier $z$ in $Z$ to
its corresponding ``semantic clone" variable. $\mathbf{S2M}$ is a mapping taking CSL's assertion to CML's assertion.
$$\mathbf{S2M}(\Gamma \vdash_{A}\{p\}k\{q\})\equiv \Gamma \vdash_{A,\{\}}\exists Z\cup Z_v((o=<\!\!k\!\!>_k<\!\!s_z\!\!>_s<\!\!-,-,-\!\!>_h<\!\!-,-\!\!>_r)\land p)$$
$$\Downarrow \exists Z\cup Z_v'((o=<\!\!\cdot\!\!>_k<\!\!s_z'\!\!>_s<\!\!-,-,-\!\!>_h<\!\!-,-\!\!>_r)\land q)$$
where ``-" is a special notation, which we call ``free-match" notation, that is, these positions
are free from match when a concrete configuration $\gamma$ matches CML pattern. For example,
let $<\!\!s\!\!>_s<\!\!h_1,h_2,H\!\!>_h<\!\!N_1,N_2\!\!>_r\in \Sigma_\Gamma$ and $\tau:\mathit{Var}\to Int$, then $$(<\!\!k\!\!>_k<\!\!s\!\!>_s<\!\!h_1,h_2,H\!\!>_h<\!\!N_1,\!N_2\!\!>_r,\tau)\!\models_{\{\}}\!\! \exists Z\cup Z_v((o\!=\!<\!\!k\!\!>_k<\!\!s_z\!\!>_s<\!\!-,-,-\!\!>_h<\!\!-,-\!\!>_r)\land p)$$
iff
\begin{itemize}
  \item there exists some $\theta_\tau:\mathit{Var}\to Int$ with $\theta_\tau\!\!\upharpoonright_{\mathit{Var}/\{Z\cup Z_v\}}=\tau\!\!\upharpoonright_{\mathit{Var}/\{Z\cup Z_v\}}$;
  \item $s=\theta_\tau(s)$;
  \item $<\!\!s\!\!>_s<\!\!h_1\!\!>_h<\!\!\{\}\!\!>_r\models p$.
\end{itemize}
\begin{thm}
If the assertion $\Gamma \vdash_{A}\{p\}k\{q\}$ is valid in CSL, then the assertion $\mathbf{S2M}(\Gamma \vdash_{A}\{p\}k\{q\})$ is valid in CML.
\end{thm}
\begin{proof}
Our semantic model is also applicable to CSL. Let $<\!\!s_c\!\!>_s<\!\!h_{1c},h_{2c},H_c\!\!>_h<\!\!N_{1c},N_{2c}\!\!>_r\in \Sigma_\Gamma$ with $\mathbf{owned}(\Gamma)\cup A \subseteq \mathbf{dom}(s_c)$ and $<\!\!s_c\!\!>_s<\!\!h_{1c},h_{2c},H_c\!\!>_h<\!\!N_{1c},N_{2c}\!\!>_r\models p$.\\
If $\Gamma \vdash_{A}\{p\}k\{q\}$ is valid, then for all traces $\alpha\in [\![k]\!]$,
\begin{itemize}
  \item $\neg(<\!\!\alpha\!\!>_k<\!\!s_c\!\!>_s<\!\!h_{1c},h_{2c},H_c\!\!>_h<\!\!N_{1c},N_{2c}\!\!>_r\xLongrightarrow[\Gamma, A, \{\}]{}^*\mathbf{abort})$;
  \item If $<\!\!\alpha\!\!>_k<\!\!s_c\!\!>_s<\!\!h_{1c},h_{2c},H_c\!\!>_h<\!\!N_{1c},N_{2c}\!\!>_r\xLongrightarrow[\Gamma, A, \{\}]{}^*<\!\!\cdot\!\!>_k<\!\!s_c'\!\!>_s<\!\!h_{1c}',h_{2c}',H_c'\!\!>_h<\!\!N_{1c}',N_{2c}'\!\!>_r$, then $<\!\!s_c'\!\!>_s<\!\!h_{1c}',h_{2c}',H_c'\!\!>_h<\!\!N_{1c}',N_{2c}'\!\!>_r\models q$.
\end{itemize}
$\mathbf{S2M}(\Gamma \vdash_{A}\{p\}k\{q\})\equiv \Gamma \vdash_{A,\{\}}\exists Z\cup Z_v((o=<\!\!k\!\!>_k<\!\!s_z\!\!>_s<\!\!-,-,-\!\!>_h<\!\!-,-\!\!>_r)\land p)\Downarrow \exists Z\cup Z_v'((o=<\!\!\cdot\!\!>_k<\!\!s_z'\!\!>_s<\!\!-,-,-\!\!>_h<\!\!-,-\!\!>_r)\land q)$.\\
It's easy to find a $\tau:\mathit{Var}\to Int$ such that $(<\!\!k\!\!>_k<\!\!s_c\!\!>_s<\!\!h_{1c},h_{2c},H_c\!\!>_h<\!\!N_{1c},N_{2c}\!\!>_r,\tau)\!\models_{\{\}}\!\! \exists Z\cup Z_v((o=<\!\!k\!\!>_k<\!\!s_z\!\!>_s<\!\!-,-,-\!\!>_h<\!\!-,-\!\!>_r)\land p)$. Since $<\!\!\alpha\!\!>_k<\!\!s_c\!\!>_s<\!\!h_{1c},h_{2c},H_c\!\!>_h<\!\!N_{1c},N_{2c}\!\!>_r\xLongrightarrow[\Gamma, A, \{\}]{}^*<\!\!\cdot\!\!>_k<\!\!s_c'\!\!>_s<\!\!h_{1c}',h_{2c}',H_c'\!\!>_h<\!\!N_{1c}',N_{2c}'\!\!>_r$ and $<\!\!s_c'\!\!>_s<\!\!h_{1c}',h_{2c}',H_c'\!\!>_h<\!\!N_{1c}',N_{2c}'\!\!>_r\models q$, then
$(<\!\!\cdot\!\!>_k<\!\!s_c'\!\!>_s<\!\!h_{1c}',h_{2c}',H_c'\!\!>_h<\!\!N_{1c}',N_{2c}'\!\!>_r,\tau)\models_{\{\}}\exists Z\cup Z_v'((o=<\!\!\cdot\!\!>_k<\!\!s_z'\!\!>_s<\!\!-,-,-\!\!>_h<\!\!-,-\!\!>_r)\land q)$.\\
Hence, the assertion $\mathbf{S2M}(\Gamma \vdash_{A}\{p\}k\{q\})$ is valid in CML.
\end{proof}
Compared with CSL, CML has the following characteristics:
\begin{itemize}
  \item The CML pattern provides more information specifications than CSL; CSL uses separation logic formula to describe the state before and after process execution. Instead of separation logic formula, CML uses pattern. Separation logic formula is relatively abstract. However, pattern involves ``low-level" operational aspects, such as how to express the state. In CML pattern, $h_1,N_1$ represents part of the heap and part of the resource owned by the process, $h_2,N_2$ is owned by the environment, $H$ represents the remaining heap and satisfies the resource invariants of the currently available resources, and the ``key set" B represents the store portion owned by the process.
  \item $\mathbf{S2M}$ is a mapping taking CSL's assertion to CML's assertion. $\mathbf{S2M}(\Gamma \vdash_{A}\{p\}k\{q\})$ is a special form of CML assertion where key set B is empty. Hence, CSL can be seen as an instance of CML.
\end{itemize}
\section{Conclusions}
Matching logic works well on the sequential processes, but not on the shared-memory concurrent processes. Nevertheless, the matching logic inherently supports the viewpoint of ``resource separation". Inspired by the concurrent separation logic (CSL), we introduce Concurrent Matching Logic (CML). However, CML's pattern involves ``low-level" operational
aspects, such as how to express the state. In addition to a ``rely set" A, we also need a ``key set" B. We give the notion
of validity and prove CML is sound to our operational semantics model for concurrency. We also analyze the relationship between the CML and the CSL, and point out that under certain assumptions, CSL can be seen as an instance of CML.
There are many extensions for CSL, such as: permissions\cite{boyland2003checking}\cite{bornat2005permission}, locks-in-the-heap\cite{gotsman2007local}\cite{hobor2008oracle}. We hope to use CML to handle these extensions, which is also our follow-up work.
\bibliographystyle{alpha}
\bibliography{lmcs}
\appendix
\section{}
\begin{proof}{(Lemma 4.1)}
Case analysis for each form of action. Most cases are straightforward. Here are the cases for $\mathbf{acq}\;r$ and $\mathbf{rel}\;r$.
\begin{itemize}
  \item For $\lambda=\mathbf{acq}\;r$; $r\notin N\cup N_3$,
  Obviously, $<\!\!\mathbf{acq}\;r\!\!>_k<\!\!s\!\!>_s<\!\!h,h_3,H\!\!>_h<\!\!N,N_3\!\!>_r\xLongrightarrow[\Gamma, A_1,B_1]{}^*\mathbf{abort}$ is vacuous. If $<\!\!\mathbf{acq}\;r\!\!>_k<\!\!s\!\!>_s<\!\!h,h_3,H\!\!>_h<\!\!N,N_3\!\!>_r\xLongrightarrow[\Gamma, A_1\cup A_2, B_1\cup B_2]{}<\!\!\cdot\!\!>_k<\!\!s'\!\!>_s<\!\!h',h_3',H'\!\!>_h<\!\!N',N_3'\!\!>_r$,
  then:
  \begin{enumerate}
    \item there are actions $\mu_1,\mu_2,\ldots,\mu_n$ such that for $1\leq i\leq n$, $\mathbf{writes}(\mu_i)\cap (A_1\cup A_2) =\emptyset$ and $\mathbf{writes}(\mu_i)\cap (B_1\cup B_2) =\emptyset$ and
  $<\!\!\mathbf{acq}\;r\!\!>_k<\!\!s\!\!>_s<\!\!h,h_3,H\!\!>_h<\!\!N,N_3\!\!>_r\rightsquigarrow_{\Gamma, A' ,B'}^*<\!\!\mathbf{acq}\;r\!\!>_k<\!\!s''\!\!>_s<\!\!h,h_3'',H''\!\!>_h<\!\!N,N_3''\!\!>_r$;
    \item there is $h_r$ such that $h_r \subseteq H''$ and $<\!\!s''\!\!>_s<\!\!h_r\!\!>_h<\!\!\{\}\!\!>_r\models \mathbf{inv}(\Gamma\upharpoonright r)$ and $<\!\!\mathbf{acq}\;r\!\!>_k<\!\!s''\!\!>_s<\!\!h,h_3'',H''\!\!>_h<\!\!N,N_3''\!\!>_r\xrightarrow[\Gamma, A_1\cup A_2,B_1\cup B_2]{}<\!\!\cdot\!\!>_k<\!\![s''\mid r\mapsto 0]\!\!>_s<\!\!h\cdot h_r,h_3'',H''-h_r\!\!>_h<\!\!N\cup\{r\},N_3''\!\!>_r$;
    \item there are $\nu_1,\nu_2,\ldots,\nu_m$ such that for $1\leq j\leq m$, $\mathbf{writes}(\nu_j)\cap (A_1\cup A_2) =\emptyset$ and $\mathbf{writes}(\nu_j)\cap (B_1\cup B_2) =\emptyset$ and $<\!\!\cdot\!\!>_k<\!\![s''\mid r\mapsto 0]\!\!>_s<\!\!h\cdot h_r,h_3'',H''-h_r\!\!>_h<\!\!N\cup\{r\},N_3''\!\!>_r\rightsquigarrow_{\Gamma, A' ,B'}^*<\!\!\cdot\!\!>_k<\!\!s'\!\!>_s<\!\!h\cdot h_r,h_3',H'\!\!>_h<\!\!N\cup\{r\},N_3'\!\!>_r$.
  \end{enumerate}
  Since $h_2\bot h_3$ and (1), then
  $$<\!\!\mathbf{acq}\;r\!\!>_k<\!\!s\!\!>_s<\!\!h_1,h_2\cdot h_3,H\!\!>_h<\!\!N_1,N_3\cup N_2\!\!>_r\rightsquigarrow_{\Gamma, A' ,B'}^*$$
  $$<\!\!\mathbf{acq}\;r\!\!>_k<\!\!s''\!\!>_s<\!\!h_1,h_2\cdot h_3'',H''\!\!>_h<\!\!N_1,N_3''\cup N_2\!\!>_r$$
 By (2)
  $$<\!\!\mathbf{acq}\;r\!\!>_k<\!\!s''\!\!>_s<\!\!h_1,h_2\cdot h_3'',H''\!\!>_h<\!\!N_1,N_3''\cup N_2\!\!>_r\xrightarrow[\Gamma, A_1,B_1]{}$$
  $$<\!\!\cdot\!\!>_k<\!\![s''\mid r\mapsto 0]\!\!>_s<\!\!h_1\cdot h_r,h_2\cdot h_3'',H''-h_r\!\!>_h<\!\!N_1\cup\{r\},N_3''\cup N_2\!\!>_r$$
  By (3)
  $$<\!\!\cdot\!\!>_k<\!\![s''\mid r\mapsto 0]\!\!>_s<\!\!h_1\cdot h_r,h_2\cdot h_3'',H''-h_r\!\!>_h<\!\!N_1\cup\{r\},N_3''\cup N_2\!\!>_r\rightsquigarrow_{\Gamma, A' ,B'}^*$$
  $$<\!\!\cdot\!\!>_k<\!\!s'\!\!>_s<\!\!h_1\cdot h_r,h_2\cdot h_3',H'\!\!>_h<\!\!N_1\cup\{r\},N_3'\cup N_2\!\!>_r$$
  Hence, $<\!\!\mathbf{acq}\;r\!\!>_k<\!\!s\!\!>_s<\!\!h_1,h_2\cdot h_3,H\!\!>_h<\!\!N_1,N_3\cup N_2\!\!>_r\xLongrightarrow[\Gamma, A_1, B_1]{}<\!\!\cdot\!\!>_k<\!\!s'\!\!>_s<\!\!h_1\cdot h_r,h_2\cdot h_3',H'\!\!>_h<\!\!N_1\cup\{r\},N_3'\cup N_2\!\!>_r$.
 The result thus holds for this case.
  \item For $\lambda=\mathbf{rel}\;r$; $r\in N$. If $<\!\!\mathbf{rel}\;r\!\!>_k<\!\!s\!\!>_s<\!\!h,h_3,H\!\!>_h<\!\!N,N_3\!\!>_r\xLongrightarrow[\Gamma, A_1\cup A_2,B_1\cup B_2]{}^*\mathbf{abort}$, then there are actions $\mu_1,\mu_2,\ldots,\mu_n$ such that for $1\leq i\leq n, 0\leq n$, $\mathbf{writes}(\mu_i)\cap (A_1\cup A_2) =\emptyset$ and $\mathbf{writes}(\mu_i)\cap (B_1\cup B_2) =\emptyset$ and
  $<\!\!\mathbf{rel}\;r\!\!>_k<\!\!s\!\!>_s<\!\!h,h_3,H\!\!>_h<\!\!N,N_3\!\!>_r\rightsquigarrow_{\Gamma, A' ,B'}^*<\!\!\mathbf{rel}\;r\!\!>_k<\!\!s''\!\!>_s<\!\!h,h_3'',H''\!\!>_h<\!\!N,N_3''\!\!>_r$ and $\forall h_r\subseteq h$, $<\!\!s''\!\!>_s<\!\!h_r\!\!>_h<\!\!\{\}\!\!>_r \models \neg\mathbf{inv}(\Gamma\upharpoonright r)$.
  Since $h=h_1\cdot h_2$, then $\forall h_r\subseteq h_1$, $<\!\!s''\!\!>_s<\!\!h_r\!\!>_h<\!\!\{\}\!\!>_r \models \neg\mathbf{inv}(\Gamma\upharpoonright r)$ and hence that
      $<\!\!\mathbf{rel}\;r\!\!>_k<\!\!s\!\!>_s<\!\!h_1,h_3\cdot h_2,H\!\!>_h<\!\!N_1,N_3\cap N_2\!\!>_r\xLongrightarrow[\Gamma, A_1,B_1]{}^*\mathbf{abort}$. \\
  If $<\!\!\mathbf{rel}\;r\!\!>_k<\!\!s\!\!>_s<\!\!h,h_3,H\!\!>_h<\!\!N,N_3\!\!>_r\xLongrightarrow[\Gamma, A_1\cup A_2, B_1\cup B_2]{}<\!\!\cdot\!\!>_k<\!\!s'\!\!>_s<\!\!h',h_3',H'\!\!>_h<\!\!N',N_3'\!\!>_r$, then
      \begin{enumerate}
    \item there are actions $\mu_1,\mu_2,\ldots,\mu_n$ such that for $1\leq i\leq n$, $\mathbf{writes}(\mu_i)\cap (A_1\cup A_2) =\emptyset$ and $\mathbf{writes}(\mu_i)\cap (B_1\cup B_2) =\emptyset$ and
  $<\!\!\mathbf{rel}\;r\!\!>_k<\!\!s\!\!>_s<\!\!h,h_3,H\!\!>_h<\!\!N,N_3\!\!>_r\rightsquigarrow_{\Gamma, A' ,B'}^*<\!\!\mathbf{rel}\;r\!\!>_k<\!\!s''\!\!>_s<\!\!h,h_3'',H''\!\!>_h<\!\!N,N_3''\!\!>_r$;
    \item there is $h_r$ such that $h_r \subseteq h$ and $<\!\!s''\!\!>_s<\!\!h_r\!\!>_h<\!\!\{\}\!\!>_r\models \mathbf{inv}(\Gamma\upharpoonright r)$ and $<\!\!\mathbf{rel}\;r\!\!>_k<\!\!s''\!\!>_s<\!\!h,h_3'',H''\!\!>_h<\!\!N,N_3''\!\!>_r\xrightarrow[\Gamma, A_1\cup A_2,B_1\cup B_2]{}<\!\!\cdot\!\!>_k<\!\![s''\mid r\mapsto 1]\!\!>_s<\!\!h-h_r,h_3'',H''\cdot h_r\!\!>_h<\!\!N-\{r\},N_3''\!\!>_r$;
    \item there are $\nu_1,\nu_2,\ldots,\nu_m$ such that for $1\leq j\leq m$, $\mathbf{writes}(\nu_j)\cap (A_1\cup A_2) =\emptyset$ and $\mathbf{writes}(\nu_j)\cap (B_1\cup B_2) =\emptyset$ and $<\!\!\cdot\!\!>_k<\!\![s''\mid r\mapsto 1]\!\!>_s<\!\!h-h_r,h_3'',H''\cdot h_r\!\!>_h<\!\!N-\{r\},N_3''\!\!>_r\rightsquigarrow_{\Gamma, A' ,B'}^*<\!\!\cdot\!\!>_k<\!\!s'\!\!>_s<\!\!h-h_r,h_3',H'\!\!>_h<\!\!N-\{r\},N_3'\!\!>_r$.
  \end{enumerate}
  Since $h_2\bot h_3$ and (1), then
  $$<\!\!\mathbf{rel}\;r\!\!>_k<\!\!s\!\!>_s<\!\!h_1,h_2\cdot h_3,H\!\!>_h<\!\!N_1,N_3\cup N_2\!\!>_r\rightsquigarrow_{\Gamma, A' ,B'}^*$$
  $$<\!\!\mathbf{rel}\;r\!\!>_k<\!\!s''\!\!>_s<\!\!h_1,h_2\cdot h_3'',H''\!\!>_h<\!\!N_1,N_3''\cup N_2\!\!>_r$$
  If $h_r\nsubseteq h_1$, then
  $<\!\!\mathbf{rel}\;r\!\!>_k<\!\!s''\!\!>_s<\!\!h_1,h_2\cdot h_3'',H''\!\!>_h<\!\!N_1,N_3''\cup N_2\!\!>_r\xrightarrow[\Gamma, A_1,B_1]{}\mathbf{abort}$  and hence that $<\!\!\mathbf{rel}\;r\!\!>_k<\!\!s\!\!>_s<\!\!h_1,h_3\cdot h_2,H\!\!>_h<\!\!N_1,N_3\cap N_2\!\!>_r\xLongrightarrow[\Gamma, A_1,B_1]{}^*\mathbf{abort}$.\\
  Otherwise, $h_r\subseteq h_1$, by (2),
  $$<\!\!\mathbf{rel}\;r\!\!>_k<\!\!s''\!\!>_s<\!\!h_1,h_2\cdot h_3'',H''\!\!>_h<\!\!N_1,N_3''\cup N_2\!\!>_r\xrightarrow[\Gamma, A_1,B_1]{}$$
  $$<\!\!\cdot\!\!>_k<\!\![s''\mid r\mapsto 1]\!\!>_s<\!\!h_1-h_r,h_2\cdot h_3'',H''\cdot h_r\!\!>_h<\!\!N_1-\{r\},N_3''\cup N_2\!\!>_r$$
  By (3)
  $$<\!\!\cdot\!\!>_k<\!\![s''\mid r\mapsto 1]\!\!>_s<\!\!h_1-h_r,h_2\cdot h_3'',H''\cdot h_r\!\!>_h<\!\!N_1-\{r\},N_3''\cup N_2\!\!>_r\rightsquigarrow_{\Gamma, A' ,B'}^*$$
  $$<\!\!\cdot\!\!>_k<\!\!s'\!\!>_s<\!\!h_1-h_r,h_2\cdot h_3',H'\!\!>_h<\!\!N_1-\{r\},N_3'\cup N_2\!\!>_r$$
  Hence, $<\!\!\mathbf{rel}\;r\!\!>_k<\!\!s\!\!>_s<\!\!h_1,h_2\cdot h_3,H\!\!>_h<\!\!N_1,N_3\cup N_2\!\!>_r\xLongrightarrow[\Gamma, A_1\cup A_2, B_1\cup B_2]{}<\!\!\cdot\!\!>_k<\!\!s'\!\!>_s<\!\!h_1-h_r,h_2\cdot h_3',H'\!\!>_h<\!\!N_1-\{r\},N_3'\cup N_2\!\!>_r$. The result thus holds for this case.
\end{itemize}
\end{proof}
\begin{proof}{(Theorem 4.3)}
By induction on the lengths of $\alpha_1$ and $\alpha_2$.
\begin{itemize}
  \item when one of the traces is empty.\\
  Without loss of generality, assume that $\alpha_2=\epsilon$, then $\alpha=\alpha_1$
  \begin{enumerate}
    \item If $<\!\!\alpha\!\!>_k<\!\!s\!\!>_s<\!\!h,h_3,H\!\!>_h<\!\!N,N_3\!\!>_r\xLongrightarrow[\Gamma, A_1\cup A_2, B_1\cup B_2]{}^*\mathbf{abort}$, then $<\!\!\alpha_1\!\!>_k<\!\!s\!\!>_s<\!\!h_1,h_3\cdot h_2,H\!\!>_h<\!\!N_1,N_3\cup N_2\!\!>_r\xLongrightarrow[\Gamma, A_1,B_1]{}^*\mathbf{abort}$ by Lemma 4.2.
    \item If $<\!\!\alpha\!\!>_k<\!\!s\!\!>_s<\!\!h,h_3,H\!\!>_h<\!\!N,N_3\!\!>_r\xLongrightarrow[\Gamma, A_1\cup A_2, B_1\cup B_2]{}^*<\!\!\cdot\!\!>_k<\!\!s'\!\!>_s<\!\!h',h_3',H'\!\!>_h<\!\!N',N_3'\!\!>_r$, by Lemma 4.2, there are $h_1'$, $h_2'=h_2$, $N_1'$, $N_2'=N_2$ such that $N_1'\cap N_2'=\emptyset$, $N'=N_1'\cup N_2'$, $h'=h_1'\cdot h_2'$ and $<\!\!\alpha_1\!\!>_k<\!\!s\!\!>_s<\!\!h_1,h_3\cdot h_2,H\!\!>_h<\!\!N_1,N_3\cup N_2\!\!>_r\xLongrightarrow[\Gamma, A_1, B_1]{}^*<\!\!\cdot\!\!>_k<\!\!s'\!\!>_s<\!\!h_1',h_3'\cdot h_2',H'\!\!>_h<\!\!N_1',N_3'\cup N_2'\!\!>_r$.
  \end{enumerate}
  The result follows.
  \item $\alpha_1=\lambda_1\alpha_1'$, $\alpha_2=\lambda_2\alpha_2'$, and $\alpha\in \alpha_{1\{A_1\}}\|_{\{A_2\}}\alpha_2$.\\
  If $<\!\!\alpha\!\!>_k<\!\!s\!\!>_s<\!\!h,h_3,H\!\!>_h<\!\!N,N_3\!\!>_r\xLongrightarrow[\Gamma, A_1\cup A_2, B_1\cup B_2]{}^*\mathbf{abort}$, then $\mathbf{writes}(\lambda_1)\cap \mathbf{free}(\lambda_2)\neq\emptyset$ or
    $(\mathbf{writes}(\lambda_2)\cap \mathbf{free}(\lambda_1)\neq\emptyset$.
    Since $\mathbf{writes}(\alpha_1)\cap A_2=\emptyset$ and $\mathbf{writes}(\alpha_2)\cap A_1=\emptyset$ and $N_1\cap N_2=\emptyset$ and $h_1\bot h_2$, it follows that either $<\!\!\lambda_1\!\!>_k<\!\!s\!\!>_s<\!\!h_1,h_3\cdot h_2,H\!\!>_h<\!\!N_1,N_3\cup N_2\!\!>_r\xLongrightarrow[\Gamma, A_1,B_1]{}^*\mathbf{abort}$ or $<\!\!\lambda_2\!\!>_k<\!\!s\!\!>_s<\!\!h_2,h_3\cdot h_1,H\!\!>_h<\!\!N_2,N_3\cup N_1\!\!>_r\xLongrightarrow[\Gamma, A_2, B_2]{}^*\mathbf{abort}$. The result then follows.\\
    Otherwise, without loss of generality, assume that: $\alpha=\lambda\alpha_3$ and  $\alpha_3\in\alpha_{1}'\;_{\{A_1\}}\|_{\{A_2\}}\mu\alpha_2'$.
    \begin{enumerate}
      \item If $<\!\!\alpha\!\!>_k<\!\!s\!\!>_s<\!\!h,h_3,H\!\!>_h<\!\!N,N_3\!\!>_r\xLongrightarrow[\Gamma, A_1\cup A_2, B_1\cup B_2]{}^*\mathbf{abort}$, then either $<\!\!\lambda\!\!>_k<\!\!s\!\!>_s<\!\!h,h_3,H\!\!>_h<\!\!N,N_3\!\!>_r\xLongrightarrow[\Gamma, A_1\cup A_2, B_1\cup B_2]{}^*\mathbf{abort}$ or there is a local state $<\!\!s''\!\!>_s<\!\!h'',h_3'',H''\!\!>_h<\!\!N'',N_3''\!\!>_r$ such that $<\!\!\lambda\!\!>_k<\!\!s\!\!>_s<\!\!h,h_3,H\!\!>_h<\!\!N,N_3\!\!>_r\xLongrightarrow[\Gamma, A_1\cup A_2, B_1\cup B_2]{}^*$\\$<\!\!\cdot\!\!>_k<\!\!s''\!\!>_s<\!\!h'',h_3'',H''\!\!>_h<\!\!N'',N_3''\!\!>_r$ and $<\!\!\alpha_3\!\!>_k<\!\!s''\!\!>_s<\!\!h'',h_3'',H''\!\!>_h<\!\!N'',N_3''\!\!>_r$\\$\xLongrightarrow[\Gamma, A_1\cup A_2, B_1\cup B_2]{}^*\mathbf{abort}$.\\
          In the first subcase, by Lemma 4.1, we get $<\!\!\lambda\!\!>_k<\!\!s\!\!>_s<\!\!h_1,h_3\cdot h_2,H\!\!>_h<\!\!N_1,N_3\cup N_2\!\!>_r\xLongrightarrow[\Gamma, A_1,B_1]{}^*\mathbf{abort}$. So $<\!\!\alpha_1\!\!>_k<\!\!s\!\!>_s<\!\!h_1,h_3\cdot h_2,H\!\!>_h<\!\!N_1,N_3\cup N_2\!\!>_r\xLongrightarrow[\Gamma, A_1,B_1]{}^*\mathbf{abort}$\\
          In the second subcase, by Lemma 4.1,\\
           either $<\!\!\lambda\!\!>_k<\!\!s\!\!>_s<\!\!h_1,h_3\cdot h_2,H\!\!>_h<\!\!N_1,N_3\cup N_2\!\!>_r\xLongrightarrow[\Gamma, A_1,B_1]{}^*\mathbf{abort}$ and  hence $<\!\!\alpha_1\!\!>_k<\!\!s\!\!>_s<\!\!h_1,h_3\cdot h_2,H\!\!>_h<\!\!N_1,N_3\cup N_2\!\!>_r\xLongrightarrow[\Gamma, A_1,B_1]{}^*\mathbf{abort}$;\\
            or there are $h_1''$, $N_1''$ such that $h''=h_1''\cdot h_2$ and $N_1''\cap N_2=\emptyset$ and $N''=N_1''\cup N_2$ and $<\!\!\lambda\!\!>_k<\!\!s\!\!>_s<\!\!h_1,h_3\cdot h_2,H\!\!>_h<\!\!N_1,N_3\cup N_2\!\!>_r\xLongrightarrow[\Gamma, A_1, B_1]{}^*<\!\!\cdot\!\!>_k<\!\!s''\!\!>_s<\!\!h_1'',h_3''\cdot h_2,H''\!\!>_h<\!\!N_1'',N_3''\cup N_2\!\!>_r$.\\
            By the induction hypothesis for $\alpha_3$, we have \\
            either $<\!\!\alpha_1'\!\!>_k<\!\!s''\!\!>_s<\!\!h_1'',h_3''\cdot h_2,H''\!\!>_h<\!\!N_1'',N_3''\cup N_2\!\!>_r\xLongrightarrow[\Gamma, A_1,B_1]{}^*\mathbf{abort}$ and hence $<\!\!\alpha_1\!\!>_k<\!\!s\!\!>_s<\!\!h_1,h_3\cdot h_2,H\!\!>_h<\!\!N_1,N_3\cup N_2\!\!>_r\xLongrightarrow[\Gamma, A_1,B_1]{}^*\mathbf{abort}$.\\
            or $<\!\!\alpha_2\!\!>_k<\!\!s''\!\!>_s<\!\!h_2,h_3''\cdot h_1'',H''\!\!>_h<\!\!N_2,N_3''\cup N''\!\!>_r\xLongrightarrow[\Gamma, A_2,B_2]{}^*\mathbf{abort}$. Since $\mathbf{writes}(\alpha_1)\cap A_2 =\emptyset$ and $\mathbf{writes}(\alpha_1)\cap B_2 =\emptyset$, by environment move, we get $<\!\!\alpha_2\!\!>_k<\!\!s\!\!>_s<\!\!h_2,h_3\cdot h_1,H\!\!>_h<\!\!N_2,N_3\cup N\!\!>_r\rightsquigarrow_{\Gamma, A_2, B_2}<\!\!\alpha_2\!\!>_k<\!\!s''\!\!>_s<\!\!h_2,h_3''\cdot h_1'',H''\!\!>_h<\!\!N_2,N_3''\cup N''\!\!>_r$ and hence $<\!\!\alpha_2\!\!>_k<\!\!s\!\!>_s<\!\!h_2,h_3\cdot h_1,H\!\!>_h<\!\!N_2,N_3\cup N\!\!>_r\xLongrightarrow[\Gamma, A_2,B_2]{}^*\mathbf{abort}$. The result then follows.
      \item If $<\!\!\alpha\!\!>_k<\!\!s\!\!>_s<\!\!h,h_3,H\!\!>_h<\!\!N,N_3\!\!>_r\xLongrightarrow[\Gamma, A_1\cup A_2, B_1\cup B_2]{}^*<\!\!\cdot\!\!>_k<\!\!s'\!\!>_s<\!\!h',h_3',H'\!\!>_h<\!\!N',N_3'\!\!>_r$, then there is a local state $<\!\!s''\!\!>_s<\!\!h'',h_3'',H''\!\!>_h<\!\!N'',N_3''\!\!>_r$ such that $<\!\!\lambda\!\!>_k<\!\!s\!\!>_s<\!\!h,h_3,H\!\!>_h<\!\!N,N_3\!\!>_r\xLongrightarrow[\Gamma, A_1\cup A_2, B_1\cup B_2]{}^*<\!\!\cdot\!\!>_k<\!\!s''\!\!>_s<\!\!h'',h_3'',H''\!\!>_h<\!\!N'',N_3''\!\!>_r$ and $<\!\!\alpha_3\!\!>_k<\!\!s''\!\!>_s<\!\!h'',h_3'',H''\!\!>_h<\!\!N'',N_3''\!\!>_r\xLongrightarrow[\Gamma, A_1\cup A_2, B_1\cup B_2]{}^*<\!\!\cdot\!\!>_k<\!\!s'\!\!>_s<\!\!h',h_3',H'\!\!>_h<\!\!N',N_3'\!\!>_r$.\\
          Use Lemma 4.1 for the first step. \\
          If $<\!\!\lambda\!\!>_k<\!\!s\!\!>_s<\!\!h_1,h_3\cdot h_2,H\!\!>_h<\!\!N_1,N_3\cup N_2\!\!>_r\xLongrightarrow[\Gamma, A_1,B_1]{}^*\mathbf{abort}$, we get $<\!\!\alpha_1\!\!>_k<\!\!s\!\!>_s<\!\!h_1,h_3\cdot h_2,H\!\!>_h<\!\!N_1,N_3\cup N_2\!\!>_r\xLongrightarrow[\Gamma, A_1,B_1]{}^*\mathbf{abort}$ as above.\\
          otherwise, there are $h_1''$, $N_1''$ such that $h''=h_1''\cdot h_2$ and $N_1''\cap N_2=\emptyset$ and $N''=N_1''\cup N_2$ and $<\!\!\lambda\!\!>_k<\!\!s\!\!>_s<\!\!h_1,h_3\cdot h_2,H\!\!>_h<\!\!N_1,N_3\cup N_2\!\!>_r\xLongrightarrow[\Gamma, A_1, B_1]{}^*<\!\!\cdot\!\!>_k<\!\!s''\!\!>_s<\!\!h_1'',h_3''\cdot h_2,H''\!\!>_h<\!\!N_1'',N_3''\cup N_2\!\!>_r$.\\
          The induction hypothesis for $\alpha_3$ implies that\\
          a), either $<\!\!\alpha_1'\!\!>_k<\!\!s''\!\!>_s<\!\!h_1'',h_3''\cdot h_2,H''\!\!>_h<\!\!N_1'',N_3''\cup N_2\!\!>_r\xLongrightarrow[\Gamma, A_1, B_1]{}^*\mathbf{abort}$, so that $<\!\!\alpha_1\!\!>_k<\!\!s\!\!>_s<\!\!h_1,h_3\cdot h_2,H\!\!>_h<\!\!N_1,N_3\cup N_2\!\!>_r\xLongrightarrow[\Gamma, A_1, B_1]{}^*\mathbf{abort}$.\\
          b), or $<\!\!\alpha_2\!\!>_k<\!\!s''\!\!>_s<\!\!h_2,h_3''\cdot h_1'',H''\!\!>_h<\!\!N_2,N_3''\cup N_1''\!\!>_r\xLongrightarrow[\Gamma, A_2, B_2]{}^*\mathbf{abort}$. we get $<\!\!\alpha_2\!\!>_k<\!\!s\!\!>_s<\!\!h_2,h_3\cdot h_1,H\!\!>_h<\!\!N_2,N_3\cup N_1\!\!>_r\xLongrightarrow[\Gamma, A_2, B_2]{}^*\mathbf{abort}$ as above.\\
          c), or there are $h_1'$, $h_2'$, $N_1'$, $N_2'$ such that $N_1'\cap N_2'=\emptyset$, $N'=N_1'\cup N_2'$, $h'=h_1'\cdot h_2'$ and
      $<\!\!\alpha_1'\!\!>_k<\!\!s''\!\!>_s<\!\!h_1'',h_3''\cdot h_2,H''\!\!>_h<\!\!N_1'',N_3''\cup N_2\!\!>_r\xLongrightarrow[\Gamma, A_1, B_1]{}^*<\!\!\cdot\!\!>_k<\!\!s'\!\!>_s<\!\!h_1',h_3'\cdot h_2',H'\!\!>_h<\!\!N_1',N_3'\cup N_2'\!\!>_r$ and $<\!\!\alpha_2\!\!>_k<\!\!s''\!\!>_s<\!\!h_2,h_3''\cdot h_1'',H''\!\!>_h<\!\!N_2,N_3''\cup N_1''\!\!>_r\xLongrightarrow[\Gamma, A_2, B_2]{}^*<\!\!\cdot\!\!>_k<\!\!s'\!\!>_s<\!\!h_2',h_3'\cdot h_1',H'\!\!>_h<\!\!N_2',N_3'\cup N_1'\!\!>_r$.\\
      Hence, $<\!\!\alpha_1\!\!>_k<\!\!s\!\!>_s<\!\!h_1,h_3\cdot h_2,H\!\!>_h<\!\!N_1,N_3\cup N_2\!\!>_r\xLongrightarrow[\Gamma, A_1, B_1]{}^*<\!\!\cdot\!\!>_k<\!\!s'\!\!>_s<\!\!h_1',h_3'\cdot h_2',H'\!\!>_h<\!\!N_1',N_3'\cup N_2'\!\!>_r$.\\
      Since $\mathbf{writes}(\alpha_1)\cap A_2 =\emptyset$ and $\mathbf{writes}(\alpha_1)\cap B_2 =\emptyset$, by environment move, we get $<\!\!\alpha_2\!\!>_k<\!\!s\!\!>_s<\!\!h_2,h_3\cdot h_1,H\!\!>_h<\!\!N_2,N_3\cup N_1\!\!>_r\xLongrightarrow[\Gamma, A_2, B_2]{}^*<\!\!\cdot\!\!>_k<\!\!s'\!\!>_s<\!\!h_2',h_3'\cdot h_1',H'\!\!>_h<\!\!N_2',N_3'\cup N_1'\!\!>_r$.
    \end{enumerate}
\end{itemize}
That completes the proof.
\end{proof}
\end{document}